%% file: main.tex
\documentclass[a4paper,UKenglish,cleveref, autoref, thm-restate]{lipics-v2021}

\pdfoutput=1 
\hideLIPIcs  

\title{The Isomorphism Problem of Power Graphs and a Question of Cameron}

\author{Bireswar Das}{IIT Gandhinagar, India}{bireswar@iitgn.ac.in}{}{}

\author{Jinia Ghosh}{IIT Gandhinagar, India}{jiniag@iitgn.ac.in}{}{}

\author{Anant Kumar}{IIT Gandhinagar, India}{kumar_anant@iitgn.ac.in}{}{}

\authorrunning{B. Das, J. Ghosh and A. Kumar} 

\Copyright{Jane Open Access and Joan R. Public} 

\ccsdesc[500]{Theory of computation~Design and analysis of algorithms}
\ccsdesc[500]{Theory of computation~Graph algorithms analysis}

\keywords{Graph Isomorphism, 
Graphs defined on Groups, 
Power Graph, 
Enhanced Power Graph, 
Directed Power Graph, 
Nilpotent Groups
} 

\category{} 

\relatedversion{} 

\funding{}

\acknowledgements{}

\nolinenumbers 

\usepackage[svgnames]{xcolor}

\usepackage{etoolbox}
\makeatletter
\patchcmd{\BR@backref}{\newblock}{\newblock($\uparrow$~}{}{}
\patchcmd{\BR@backref}{\par}{)\par}{}{}
\makeatother

\hypersetup{
    colorlinks=true,
    linkcolor=IndianRed,
    filecolor=cyan,      
    urlcolor=DodgerBlue,
    citecolor = SeaGreen
}

\usepackage{charter}
\usepackage{parskip}

\usepackage{wrapfig}

\usepackage{tikz}
\usetikzlibrary{positioning}
\usetikzlibrary{arrows}
\usetikzlibrary{decorations}
\usetikzlibrary{shapes.geometric}
\usepackage{algorithm}
\usepackage{hyperref}
\usepackage{algpseudocode}  
\usepackage{latexsym,amssymb,amsmath,mathtools,eulervm,xspace,nicefrac}
\usepackage[color=blue!20,textsize=small,textwidth=2.2cm]{todonotes}
\usepackage{tcolorbox}
\usepackage{adjustbox}

\EventEditors{John Q. Open and Joan R. Access}
\EventNoEds{2}
\EventLongTitle{42nd Conference on Very Important Topics (CVIT 2016)}
\EventShortTitle{CVIT 2016}
\EventAcronym{CVIT}
\EventYear{2016}
\EventDate{December 24--27, 2016}
\EventLocation{Little Whinging, United Kingdom}
\EventLogo{}
\SeriesVolume{42}
\ArticleNo{23}

\begin{document}

\maketitle
\begin{abstract}
The isomorphism problem for graphs (GI) and the isomorphism problem for groups (GrISO) have been studied extensively by researchers. The current best algorithms for both these problems run in quasipolynomial time. In this paper, we study the isomorphism problem of graphs that are defined in terms of groups, namely power graphs, directed power graphs,  and enhanced power graphs. It is not enough to check the isomorphism of the underlying groups to solve the isomorphism problem of such graphs as the power graphs (or the directed power graphs or the enhanced power graphs) of two nonisomorphic groups can be isomorphic. Nevertheless, it is interesting to ask if the underlying group structure can be exploited to design better isomorphism algorithms for these graphs. We design polynomial time algorithms for the isomorphism problems for the power graphs, the directed power graphs and the enhanced power graphs arising from finite nilpotent groups. In contrast, no polynomial time algorithm is known for the group isomorphism problem, even for nilpotent groups of class 2. 

We note that our algorithm does not require the underlying groups of the input graphs to be given. The isomorphism problems of power graphs and enhanced power graphs are solved by first computing the directed power graphs from the input graphs. The problem of efficiently computing the directed power graph from a power graph or an enhanced power graph is due to Cameron [IJGT'22]. Recently, Bubboloni and  Pinzauti [Arxiv'22] gave a polynomial time algorithm to reconstruct the directed power graph from a power graph. We give an efficient algorithm to compute the directed power graph from an enhanced power graph. The tools and techniques that we design are general enough to give a different algorithm to compute the directed power graph from a power graph as well.

\end{abstract}
\section{Introduction}

Given two graphs as input, the graph isomorphism problem (GI) is to check if the graphs are isomorphic. Despite extensive research, the complexity status of the graph isomorphism problem is still open. The best-known algorithm for GI is due to Babai, and the runtime of the algorithm is quasipolynomial \cite{babai2016graph}. The graph isomorphism problem is in NP but it is very unlikely to be NP-hard as the problem is also in coAM \cite{boppana1987does}.

Efficient algorithms for the graph isomorphism problem are known for several restricted graph classes, for example, graphs with bounded genus 
\cite{miller1980isomorphism, grohe2019linear}, graphs with bounded degree \cite{luks1982isomorphism}, graphs with bounded eigenvalue multiplicity \cite{babai1982isomorphism}, graphs with bounded tree-width \cite{bodlaender1990polynomial} and with bounded rank-width \cite{grohe2021isomorphism}.

Group theoretic tools have played an important role in the design of efficient algorithms for the graph isomorphism problem. Some of the early works of using the structure of groups non-trivially include the isomorphism algorithm for bounded degree graphs by Luks \cite{luks1982isomorphism} and the graph canonization framework developed by Babai and Luks \cite{babai1983canonical}.
Babai developed sophisticated new  techniques to give a quasipolynomial time isomorphism algorithm \cite{babai2016graph}. Group theoretic machinery has been used to design faster isomorphism algorithms for bounded degree graphs by Grohe, Neuen and Schweitzer \cite{grohe2020faster}; for graphs with bounded tree-width by Grohe, Neuen, Schweitzer, Wiebking \cite{grohe2020improved} and Wiebking \cite{wiebking2020graph};
and for bounded rankwidth graph by Grohe and Schweitzer \cite{grohe2015isomorphism} etc.

 In this paper, we study the isomorphism problem of graphs defined on finite groups. More precisely, we study the class of power graphs, directed power graphs, and enhanced power graphs. For two elements $x$ and $y$ in a finite group $G$, we say that $y$ is a power $x$ if $y=x^i$ for some integer $i$. For a group $G$, the vertex set of the \emph{power graph} $Pow(G)$ of $G$ consists of elements of $G$. Two vertices $x$ and $y$ are adjacent in $Pow(G)$ if $x$ is a power of $y$ or $y$ is a power of $x$. We refer to $G$ as the \emph{underlying group} of $Pow(G)$. The definition of directed power graphs and enhanced power graphs can be found in \Cref{prelim}.

Kelarev and Quinn defined the concept of  directed power graphs of semigroups \cite{kelarev2000combinatorial}. Power graphs were defined by Chakrabarty et al. \cite{chakrabarty2009undirected} again for semigroups. The paper by Cameron \cite{cameron2021graphs} discusses several graph classes defined in terms of groups and surveys many interesting results on these graphs. Kumar et al. \cite{kumar2021recent} gave a survey on the power graphs of finite groups. Questions related to isomorphism of graphs defined on groups have also been studied \cite{arvind2022recognizing,feng2016full}.


Our motivation for studying the isomorphism of graphs defined in terms of groups is to explore if the group structure can be exploited to give efficient algorithms for the isomorphism problems of these graphs. There are two versions of the isomorphism problem for each class of graphs defined on groups. For example, let us consider the case for the class of power graphs. In the first version of the problem, two groups $G_1$ and $G_2$ are given by their Cayley tables, and the task is to check if $Pow(G_1)$ is isomorphic to $Pow(G_2)$. In the second version, two power graphs $\Gamma_1$ and $\Gamma_2$ are given and we need to check if $\Gamma_1$ is isomorphic to $\Gamma_2$. 

In the first version of the problem, it is tempting to use the isomorphism of the underlying groups in the hope that it might yield an easier \footnote{compared to Babai's quasipolynomial time isomorphism algorithm.} quasipolynomial time algorithm because, unlike graphs, the quasipolynomial time algorithm for groups attributed to Tarjan by Miller \cite{miller1978nlog} is much easier. However, we note that it is not enough to check the isomorphism of the underlying groups, as two nonisomorphic groups can have isomorphic power graphs\footnote{ To see this, we can take the elementary abelian group of order 27 and the non-abelian group of order 27 with exponent 3 (\cite{cameron2011power}). In general, consider the power graphs of two nonisomorphic groups of order $p^i$ for any  $i\geq 2$ and exponent $p$ for some prime $p$. One can check that the power graphs are isomorphic while the groups are not.}.

The second version looks more challenging since we do not have the underlying groups. In this paper, we show that the isomorphism problem of power graphs of nilpotent groups can be tested in polynomial time even in the second version of the isomorphism problem (see \Cref{section 6 ISO}). Thus, we do not need the underlying groups to be given. 
In contrast, the group isomorphism problem for nilpotent groups, even for class 2, is still unresolved and is considered a hard instance for the group isomorphism problem \cite{grochow2021p}.

This paper makes contributions in three algebraic and combinatorial techniques, which form the foundation of our algorithms. Firstly, we introduce a simple yet effective concept of certain type of generating sets of a group which we call covering cycle generating sets (CCG-set). In essence, these are defined in terms of the set of maximal cyclic groups (refer to \cref{cyclic cover section} for details). These CCG-sets provide a crucial framework for our work. Secondly, we present a set of findings concerning the structure of closed-twins within specific subgraphs of power graphs. While Cameron previously explored the structure of closed-twins in a power graph   \cite{cameron2010power}, we extend this investigation to focus on the subgraph induced by the closed neighbourhood of a vertex (Section \ref{section 4 twin structure}). These structures form the basis of our algorithm to determine whether a vertex in a power graph can be part of a CCG-set. Lastly, we introduce a series of reduction rules that facilitate the simplification of the structure of a directed power graph while preserving its isomorphism-invariant properties (Section \ref{section 6 ISO}). These reduction rules serve as a valuable tool in streamlining the analysis on directed power graphs.


Our algorithm for solving the isomorphism problem of power graphs works by first computing the directed power graphs of the input power graphs. Next, we use the algorithm for the isomorphism problem of directed power graphs for nilpotent groups that we design in \Cref{section 6 ISO}.

The question of efficiently computing the directed power graph from the power graph (or the enhanced power graph) was asked by Cameron \cite{cameron2021graphs}: ``Question 2: Is there a simple algorithm for constructing the directed power graph or the enhanced power graph from the power graph, or the directed power graph from the enhanced power graph?'' 
Recently, Bubboloni and  Pinzauti \cite{bubboloni2022critical}  gave a polynomial time algorithm to reconstruct the directed power graph from the power graph. We give an efficient algorithm to compute the directed power graph from the enhanced power graph. Using CCG-sets, our results on the structure of closed-twins, and the reduction rules, we also give a different solution to the problem of efficiently computing the directed power graph from the power graph.   
 
Cameron \cite{cameron2010power}, and Zahirovi\'c et al. \cite{zahirovic2020study} proved that for two finite groups, the power graphs are isomorphic if and only if the directed power graphs of the groups are isomorphic if and only if the enhanced power graphs of the groups are isomorphic. The algorithms to solve Cameron's question provide a complete algorithmic proof of this result.

\emph{Related work on Cameron's question}: The algorithm to reconstruct the directed power graph from a power graph by Bubboloni and Pinzauti works by first considering the notion of \emph{plain} and \emph{compound} closed-twin classes in a power graph\footnote{These notions are similar to the closed-twin classes that Cameron calls type (a) and type (b) (Proposition 5, \cite{cameron2010power}).}. The other notion is that of \emph{critical} closed-twin classes. Depending on whether a closed-twin class is critical or non-critical they design efficient tests to determine if the class is plain or compound. Once this is done, they put directions to the edges of the power graph to reconstruct the directed power graph. The details can be found in \cite{bubboloni2022critical}. In contrast, our algorithm identifies a CCG-set in the power graph by using the properties and algorithms associated with the graph reductions that we give in this paper.

\section{Preliminaries}
\label{prelim}

 For a simple graph $X$, $V(X)$ denotes the vertex set of $X$ and $E(X)$ denotes the edge set of $X$. For basic definitions and notations from graph theory, an interested reader can refer to any standard textbook (for example, \cite{west2001introduction}). A \textit{subgraph} of $X$ is a graph $Y$, where $V(Y)\subseteq V(X)$ and $E(Y)\subseteq E(X)$. Let $S \subseteq V(X)$. Then the subgraph with the vertex set $S$ and all such edges in $E(X)$ whose both endpoints are in $S$ is called the \textit{induced subgraph} of $X$ on $S$ and it is denoted by $X[S]$.

The set of vertices adjacent to a vertex $u$ in an undirected graph $X$ is called the open neighborhood of $u$ in $X$ and is denoted by $N_X(u)$. The cardinality of $N_X(u)$ is called the \textit{degree} of $u$ in $X$, denoted by $deg_X(u)$. The \textit{closed neighborhood} of a vertex $u$ in $X$ is denoted by $N_X[u]$ and defined by $N_X[u]=N_X(u)\cup \{u\}$. Two vertices in $X$ are called the \emph{closed-twins} in $X$ if their closed neighborhoods in $X$ are the same. 

For a directed graph $X$ (with no multiple edges), the \textit{out-neighborhood} of a vertex $u$ in $X$ is the set $\{v \in V(X) \ :\ (u,v)\ \in E(X) \}$ and \textit{out-degree} of $u$ in $X$, denoted by $out$-$deg_{X}(u)$, is the size of the out-neighborhood of $u$ in $X$. Similarly, the \textit{in-neighborhood} of a vertex $u$ in $X$ is the set $\{v\in V(X) \ :\ (v,u)\ \in E(X) \}$ and \textit{in-degree} of $u$ in $X$, denoted by $in$-$deg_{X}(u)$, is the size of the in-neighborhood of $u$ in $X$.\footnote{When the graph is clear from the context, we drop the suffixes.} Two vertices in a directed graph $X$ are called the \emph{closed-twins} in $X$ if their closed-out-neighborhoods in $X$ are the same and also the closed-in-neighborhoods in $X$ are the same. An edge of the form $(u,u)$ in a directed graph is called a \textit{self-loop}.

In any graph $X$, the \emph{closed-twin-class} of a vertex $u$ in $X$ is the set of all closed-twins of $u$ in $X$. 

\begin{definition}
    Two graphs $X$ and $Y$ are called \emph{isomorphic} if and only if there exists a bijection $f$ from $V(X)$ to $V(Y)$ such that $\{u,v\}\in E(X)$ if and only if $\{f(u),f(v)\}\in E(Y)$. Moreover, if $X$ and $Y$ are vertex-colored, then an isomorphism $f$ is called a \emph{color preserving isomorphism} if for all $u \in V(X)$, the color of $u$ and the color of $f(u)$ are the same.
\end{definition}
In this paper, if the underlying graphs are colored, then by isomorphism we mean color preserving isomorphism only.

\begin{definition}\label{Definition:strong product} (see for example \cite{imrich2008topics})
Let  $X$ and $Y$ be two directed graphs. The \emph{strong product} ($X\boxtimes Y$) of $X$ and $Y$ is the graph with the vertex set $V(X)\times V(Y)$, where there is an edge from $(u,u')$ to a distinct vertex $(v,v')$ in $X\boxtimes Y$ if and only if one of the following holds:
\begin{enumerate}
    \item $u= v$ and there is an edge from $u'$ to $v'$ in $Y$.
    \item $u'= v'$ and there is an edge from $u$ is to $v$ in $X$.
    \item There is an edge from $u$ to $v$ in $X$ and an edge from $u'$ to $v'$ in $Y$.
\end{enumerate}
\end{definition}

A graph is called \textit{prime graph}, if it cannot be represented as a strong product of two non-trivial graphs.

\begin{definition}(see for example \cite{oxley2006matroid})
  Vertex identification of a pair of vertices $v_1$ and $v_2$ of a graph is the operation that produces a graph in which the two vertices $v_1$ and $v_2$ are replaced with a new vertex $v$ such that $v$ is adjacent to the union of the vertices to which $v_1$ and $v_2$ were originally adjacent. In vertex identification, it doesn't matter whether $v_1$ and $v_2$ are connected by an edge or not. 
\end{definition}
The basic definitions and properties from group theory can be found in any standard book (see, for example, \cite{rotman2012introduction}). All the groups considered in this paper are finite. A subset $H$ of a group $G$ is called a \textit{subgroup} of $G$ if $H$ forms a group under the binary operation of $G$; it is denoted by $H \leq G$.

The number of elements in $G$ is called the \textit{order} of the group and it is denoted by $|G|$. The \textit{order} of an element $g$ in $G$, denoted by $o(g)$, is the smallest positive integer $m$ such that $g^m=e$, where $e$ is the identity element. The set $\{g,g^2,g^3,\dots,g^{m-1},e\}$ is the set of all group elements that are \textit{generated} by $g$, where $m=o(g)$. Moreover, this set forms a subgroup of $G$ and is called the \textit{cyclic subgroup} generated by $g$ and denoted by $\langle g \rangle$. The number of generators of a cyclic subgroup $\langle g \rangle$ is $\phi(o(g))$, where $\phi$ is the Euler's totient function. A group $G$ is called \textit{cyclic} if $G=\langle g \rangle$, for some $g \in G$. In a finite cyclic group $G$, for any factor $m$ of $|G|$, $G$ has a unique subgroup of order $m$. This is known as the converse of Lagrange's theorem for finite cyclic groups.

A group $G$ is called a \textit{$p$-group} if the order of each element is some power of $p$, where $p$ is a prime.
For a prime $p$, if $p^m$ is the highest power of $p$ such that $p^m$ divides $|G|$, then a subgroup $H \leq G$ with the property $|H|=p^m$ is called a \textit{Sylow p-subgroup} of $G$. The \textit{direct product} of two groups $G$ and $H$, denoted by $G\times H$, is the group with elements  $(g,h)$ where  $g \in G$ and $h \in H$. The group operation of $G\times H$ is given by $(g_1,h_1)(g_2,h_2)=(g_1g_2,h_1h_2)$, where the co-ordinate wise operations are the group operations of $G$ and $H$ respectively. A finite group is called a \textit{nilpotent group} if it is a direct product of its Sylow p-subgroups.

We now give the definitions of graphs defined on groups that we discuss in this paper (see \cite{cameron2021graphs}). 
\begin{definition}
The directed power graph of a group $G$, denoted $DPow(G)$, is a directed graph with vertex set $G$, and edge set $E=  \{(x,y): y = x^m$ \textrm{ for some integer} m \}. 
\end{definition}
If $(x,y)$ is an edge in $DPow(G)$, then $o(y)$ divides $o(x)$. Let $\mathscr{DP}ow$ denote the set $\{DPow(G)\ : \ G \textrm{ is a finite group }\}$.


\begin{definition}
The power graph of a group $G$, denoted by $Pow(G)$, is the undirected version $DPow(G)$.
\end{definition}
 If $\{x,y\}$ is an edge in $Pow(G)$, then $o(x)|o(y)$ or $o(y)|o(x)$. Let $\mathscr{P}ow$ denote the set $\{Pow(G)\ : \ G \textrm{ is a finite group }\}$.

\begin{definition}
The enhanced power graph of a group $G$, denoted $EPow(G)$, is a graph with vertex set $G$, in which two vertices $x$ and $y$ are adjacent if and only if they are in the same cyclic subgroup of $G$, i.e., there exists $z$ in $G$ such that $x,y \in \langle z \rangle$. 
\end{definition}
Let $\mathscr{EP}ow$ denote the set $\{EPow(G)\ : \ G \textrm{ is a finite group }\}$.


\section{Cyclic cover of a group and its properties}
\label{cyclic cover section}

In this section we introduce the notion of minimal cyclic cover and  covering cycle generating set. We start with following definition.
\begin{definition}
    We say that a proper cyclic subgroup $C$ of $G$ is a \emph{maximal cyclic subgroup} if for all cyclic subgroups $C'$, $C \leq C'$ implies $C=C'$ or $C'=G$.
\end{definition}
\begin{definition}
    Let $G$ be a finite group. Let $C_1, C_2,\ldots, C_m$ be a set of the cyclic subgroups of $G$. We say that $C_1,C_2,\ldots ,C_m$ is a \emph{minimal cyclic cover} if $G=\cup_{i=1}^m C_i$ and $\cup_{i\neq j} C_i\neq G$ for all $j=1,\ldots ,m$.
\end{definition}
\begin{lemma}\label{unique minimum cyclic cover}
    If $G$ is a cyclic group, then $C=G$ is the only minimal cyclic cover; otherwise, the set of all maximal cyclic subgroups of $G$ forms the unique minimal cyclic cover.
\end{lemma}
\begin{proof}\label{proof:unique minimum cyclic cover}
    The case when $G$ is a cyclic group is easy. Assume that $G$ is not cyclic. Then the set of all maximal cyclic subgroups $\{C_1,\ldots ,C_m$\} of $G$ is non-empty. Since for all $g$ there is a maximal cyclic subgroup containing $g$, we must have $G=C_1\cup \ldots \cup C_m$. If $G=\cup_{i\neq j} C_i$ for some $j$, then any generator of $C_j$ is in one of the cyclic groups $C_1,C_2,\ldots, C_{j-1},C_{j+1},\ldots, C_m$. However, this is not possible. So, $\{C_1,\ldots, C_m\}$ is a minimal cyclic cover.

    Suppose $\{D_1,\ldots, D_k\}$ is a minimal cyclic cover. Now a generator of a maximal cyclic subgroup $C_i$ is in one of the $D_j's$. This forces $D_j=C_i$. Thus, each $D_j$ is a maximal cyclic subgroup. However, we have seen a proper subset of $\{C_1,\ldots, C_m\}$ cannot cover $G$. Therefore, $\{D_1,\ldots, D_k\}=\{C_1,\ldots, C_m\}$.
\end{proof}

Following the above lemma, we can see that the set of all maximal cyclic subgroups forms the minimum cyclic cover of a non-cyclic group.

\begin{definition}
Let $\{C_1,\ldots, C_m\}$ be the minimum cyclic cover of $G$. A cyclic group $C_i$ in the minimum cyclic cover is called a \emph{covering cycle}. For a cyclic group $C$, let $gen(C)$ be the set of generators of $C$. An element in $\cup_{i=1}^m gen(C_i)$ is called a covering cycle generator or CC-generator. We call a set $\{g_1,g_2,\dots,g_m\}$ a \emph{covering cycle generating set} (CCG-set) if $\{\langle g_1\rangle,\langle g_2\rangle, \dots,\langle g_m\rangle \} = \{ C_1,C_2,\dots,C_m\} $.
\end{definition}
The above definition includes the case when $m=1$, i.e., $G$ is cyclic.

~

\noindent
{\bf Organization of the paper:} With the notion of CCG-set defined above, we are now ready to describe the organization of the paper. For the sake of clarity we give the organization in a somewhat nonlinear manner.

How to identify a CCG-set in a power graph or in an enhanced power graph when the underlying group is not given? We design algorithms in Section \ref{ccg from a graph} to solve this problem. These are iterative algorithms that take one of the potential vertices and decide if that vertex can be safely marked as a member of the CCG-set.

The correctness of the algorithm, in the case of power graphs, crucially depends on the structure  of closed-twins in the subgraph induced by the closed neighbourhood of the potential vertex that the algorithm examines in each iterative step. In Section \ref{section 4 twin structure}, we derive a collection of results that characterizes these structures.

In Section~\ref{section 6 ISO}, we define a set of reduction rules that simplifies the structure of the directed power graph of a group $G$ while retaining all its isomorphism-invariant properties. There are four reductions, Reduction 1, 2, 3, and 4, and they are applied one after the other. The graph obtained after $i$-th reduction is denoted by $R_i(G)$. We also show that these reductions can be efficiently reversed. The process of obtaining $R_3(G)$ from $R_4(G)$ is done using Algorithm \ref{algorithm:R4 to R3}. The proof of correctness of Algorithm~\ref{algorithm:R4 to R3} is technical and has some nuances. This proof is presented in Section~\ref{Proof of Lemma 39}.

We design an isomorphism algorithm for the directed power graphs of nilpotent groups using the structure of $R_3$ in Section~\ref{section 6 ISO}.

In Section~\ref{section 7 Q2}, we show how to obtain the reduced graph $R_4$ from an input power graph (or an enhanced power) graph given along with a CCG-set.

Combining the above, we have an efficient way of going from a power graph (or an enhanced power graph) to $R_4$ to $R_3$ to the directed power graph. This answers  Cameron's question positively.

The isomorphism of the power graphs (or the enhanced power graphs) of nilpotent groups is now straightforward: we just need to compute the directed power graphs and apply the algorithm developed in Section~\ref{section 6 ISO}.

\section{Structure of closed-twins in a power graph}\label{section 4 twin structure}

The structure of closed-twins in a power graph has been studied by Cameron \cite{cameron2010power},  and by Bubboloni and Pinzauti \cite{bubboloni2022critical}.
In this section, we explore the structures of closed-twins in the subgraph of a power graph induced by the closed neighborhood of a vertex.  
We show in \Cref{section: CCG of power graph} that these structures can be used to find a CCG-set of a group from the corresponding power graph, even when the group is not given.

First, we note an easy fact about the closed-twins in any graph.
\begin{lemma}
\label{twins remain twin}
    Let $X$ be a graph and let $v\in V(X)$. Suppose $x$ and $y$ are closed-twins in $X$. If $x \in N[v]$, then $y \in N[v]$. Moreover, $x$ and $y$ are closed-twins in $X[N[v]]$.
\end{lemma}

Let $G$ be a group. It is easy to see that an 
 element $x\in G$ and any generator of $\langle x \rangle $ are closed-twins in $\Gamma=Pow(G)$. Therefore applying \Cref{twins remain twin}, we have the following corollary: 

\begin{corollary}
\label{generators and twins}
    Let $v\in V(\Gamma)$. If $x \in N[v],$ then all the generators of $\langle x \rangle $ are in $N[v]$. Moreover, they are closed-twins in $\Gamma_v$, where $\Gamma_v=\Gamma[N[v]]$.
\end{corollary}

Now consider a vertex $v \in V(\Gamma)$, where $\Gamma \in \mathscr{P}ow$ and the subgraph $\Gamma_v=\Gamma[N[v]]$ induced on the closed neighborhood of $v$. For any vertex $x$ in $\Gamma_v$, $o(x)|o(v)$ or $o(v)|o(x)$. We partition $V(\Gamma_v)$ according to the order 
of the vertices in the following way:
\begin{center}
    
        $U_v=\{x\in V(\Gamma_v): o(x) > o(v)\}$\\
        $E_v=\{x\in V(\Gamma_v): o(x) = o(v) \}$\\
        $L_v=\{x\in V(\Gamma_v): o(x) < o(v) \}$
    
\end{center}
For a vertex $x\in U_v$, we have $o(v)|o(x)$ and for a vertex $x\in L_v$, we have $o(x)|o(v)$.
\begin{definition}
    For a prime $p$, we say that an element in a group is a $p$-power element if $o(x)=p^i$ for some $i\geq0$. We say that $x$ is a nontrivial $p$-power element if $o(x)=p^i$ for some $i>0$.\\
\end{definition}

\begin{lemma}
\label{twin in U has more degree}
    Suppose $v \in V(\Gamma)$ is not a $p$-power element and $x \in U_v$ is a closed-twin of $v$ in $\Gamma_v$. Then, $deg_{\Gamma}(x) > deg_{\Gamma}(v)$.
\end{lemma}
\begin{proof}
    In this case, there exists prime $q$ and positive integer $s$ such that $q^s | o(x)$ but $q^s \nmid o(v)$. Then, $x$ has a neighbor $z=x^{\frac{o(x)}{q^s}}$ of order $q^s$ (by the converse of Lagrange's theorem in finite cyclic groups). Note that $z$ is not a neighbor of $v$ as $o(z)\nmid o(v)$ and also $o(v)\nmid o(z)$. The latter is true as $o(v)$ is divisible by at least two distinct primes. 
\end{proof}


\begin{lemma}
\label{generators are the only twins}
    Let $v\in V(\Gamma)$ be a  CC-generator such that $o(v)$ is  not a prime power. Let $u \in V(\Gamma_v)$.
If $u = e$ or $u$ is a generator of $\langle v \rangle,$ then the closed-twins of $u$ in $\Gamma_v$ are exactly the generators of $\langle v \rangle $ and $e$; otherwise, the closed-twins of $u$ in $\Gamma_v$ are exactly the generators of $\langle u \rangle $.
  
\end{lemma}
\begin{proof}
\label{Proof: generators are the only twins }

Let $o(v)=p_1^{r_1}p_2^{r_2}\dots p_k^{r_k},$ where $k \geq 2$. The case when $u=e$ or $u$ is a generator of $\langle v \rangle $ is easy as $N[u]=V(\Gamma_v)$ for any such element. Otherwise, since $v$ is a CC-generator, $\langle u \rangle  \lneq \langle v \rangle $.
    For $u$ and $z$ to be closed-twins, we must have $u \in \langle z \rangle $ or $z \in \langle u \rangle $. We show that for $z$ to be a closed-twin of $u$, its order must be the same as that of $u$. We consider the case when $z \in \langle u \rangle$. The other case can be handled similarly. In this case, we have $o(z)|o(u)$.

    Suppose both $u$ and $z$ are $p$-power elements for some prime $p \in \{p_1,p_2, \dots ,p_k\}$. Moreover, without loss of generality, assume that $o(u)=p_1^{s_1}$ and  $ o(z)=p_1^{s'_1}$ where $s_1 > s'_1$. Note that $r_1 \geq s_1$. In this case, there is an element in $V(\Gamma_v)$ of order $p_1^{s'_1}p_2$ which is adjacent to $z$ but not to $u$. More precisely, this element is an element in $\langle v \rangle $ of order $p_1^{s'_1}p_2$. So, in this case, $u$ and $z$ are not closed-twins in $\Gamma_v$.

Now suppose $o(u)$ is not a prime power. We first take $z$ to be non-identity. Then, let $o(u)=p_1^{s_1}\dots p_k^{s_k}$, where $k\geq 2$. Let $o(z)=p_1^{s'_1}\dots p_k^{s'_k}$, where $s_j \geq s'_j$. Assume without loss of generality that $s_1 > s'_1$. As $o(u)$ is not a prime power order, we can take $s_2\neq 0$. Now if $s'_2=0,$ consider an element $x$ of order $p_2$ in $\Gamma_v$. Then $x$ is a neighbor of $u$, but not of $z$. On the other hand, if $s'_2 \neq 0$, we take an element $y$ of order $p_1^{s_1}$. Again $y$ is a neighbor of $u$ but not of $z$. So, in this case also, $u$ and $z$ are not closed-twins in $\Gamma_v$. 

Now suppose $o(u)$ is not a prime power, i.e., $o(u)=p_1^{s_1}p_2^{s_2}\dots p_k^{s_k}$ where $k\geq 2$ and $z$ is identity. We recall that since $u$ is not a generator $\langle v \rangle$, there exists $i$ such that $r_i>s_i$. We take an element $x$ of order $p_i^{r_i}$ in $\Gamma_v$. One can check that $x$ is adjacent to $z$ but not to $u$. So, here also $u$ and $z$ are not closed-twins in $\Gamma_v$.

Note that if $o(u)=o(z),$ then they are closed-twins in $\Gamma_v$.
\end{proof}

\begin{remark}
\label{phi}
    If $a|b$, $a\neq b$, then (1) $\phi(a)|\phi(b)$, (2) $\phi(a)\leq \phi(b)$ and the equality holds only when $b=2a$ where $a$ is an odd natural number.
\end{remark}


If $v\in V(\Gamma)$ is a CC-generator, it is easy to see that $o(v)=deg(v)+1=|\Gamma_v|$. Let $o(v)$ be not a prime power. Then using \Cref{generators are the only twins}, the set of dominating vertices in $\Gamma_v$ is the set of generators of $\langle v \rangle$ and identity. Thus, the size of the closed-twin-class of $v$ in $\Gamma_v$ is $\phi(o(v))+1$, i.e., $\phi(|\Gamma_v|)+1$. Also, for all divisors, $1 < k < o(v)$, of $o(v)$, there exists a closed-twin-class of size $\phi(o(k))$ in $\Gamma_v$. Therefore, using \Cref{generators are the only twins} and \Cref{phi}, we have the following corollary:

\begin{restatable}{corollary}{corollaryTwinSize}\label{corollary twin size}
     Let $v\in V(\Gamma)$ be a  CC-generator and $o(v)$ be not a prime power. Then the following holds:
 \begin{enumerate}
     \item The size of the closed-twin-class of $v$ in $\Gamma_v$, i.e., the set of dominating vertices in $\Gamma_v$, is $\phi(o(v))+1$.
     \item For each divisor $k$ of $o(v)$, $1<k <o(v)$, there is a closed-twin-class of size $\phi(k)$ in $\Gamma_v$. Moreover, $\phi(k)$ divides $\phi(o(v))$.
     \item There are at most two closed-twin-classes of size greater or equal to $\phi(o(v))$.
 \end{enumerate}   
\end{restatable}

\begin{proof}
        \begin{enumerate}
        \item  
        Using \Cref{generators are the only twins}, the closed-twins of $v$ in $\Gamma_v$ are the generators of $<v>$ and identity. This proves our claim. Also, note that these vertices form the set of dominating vertices in $\Gamma_v$.
        \item By converse of Lagrange's theorem for finite cyclic groups, we know that for each divisor $k$ of $o(v)$, there exists a unique cyclic subgroup of order $k$. Also, for each such $k$ there is a closed-twin-class of size $\phi(k)$, due to \Cref{generators are the only twins}. Moreover, using \Cref{phi} we can conclude $\phi(k)|\phi(o(v))$.
        
        \item From (1) and (2), we know that the size of a closed-twin-class in $\Gamma_v$ is either $\phi(o(v))+1$ or $\phi(k)$, where $k$, $1<k <o(v)$, is a divisor of $o(v)$. Now using \Cref{phi}, we can see that $\phi(o(v))= \phi(k)$ if and only if $o(v)=2\cdot k$ where $k$ is odd. Thus, there can be at most two closed-twin-classes of size greater than or equal to $\phi(o(v))$.

    \end{enumerate}

\end{proof}

The following theorem is a well-known result  \cite{chakrabarty2009undirected}. We give a proof for the sake of completeness in \Cref{Proof of Theorem 18}.

\begin{restatable}{theorem}{powCliquelemma}{\cite{chakrabarty2009undirected}}
\label{complete power graph}
    Let $G$ be a finite group. Then, $\Gamma =Pow(G)$ is complete if and only if $G$ is cyclic of prime power order.
\end{restatable}

From the above theorem, the following corollary is immediate.
\begin{corollary}\label{lemma 1a}
    Let $v\in V(\Gamma)$ be a $p$-power element for some prime $p$. Then, $\Gamma[E_v \cup L_v]$ is a complete graph. Moreover, the elements of $E_v \cup L_v$ are closed-twins of $v$ in $\Gamma_v$.
\end{corollary}

\begin{lemma}
\label{v prime power order}
    Let $v\in V(\Gamma)$ be a nontrivial $p$-power element and not a CC-generator. Suppose for all $u \in U_v$ such that $u$ is a closed-twin of $v$ in $\Gamma_v$, $deg_{\Gamma}(u)$ is at most $deg_{\Gamma}(v)$. Let $y$ be a closed-twin of $v$ in $\Gamma_v$ with maximum order and $S$ denotes the set $\{ x\in V(\Gamma_v): o(y)|o(x) \text{ and $o(x)\neq o(y)$}\}$. Then,
    \begin{enumerate}
        \item The closed-twins of $v$ are exactly the elements in $\langle y \rangle $.
        \item $V(\Gamma_v)=\langle y \rangle \sqcup  S ,$ where $\sqcup$ denotes the disjoint union.
        \item Moreover, if $o(y)=p^j$ where $j\geq 2$, then $p$ divides $\lvert V(\Gamma_v) \rvert$.
    \end{enumerate}
\end{lemma}
\begin{proof}

Suppose $o(v)=p^i$. From \Cref{lemma 1a}, we know that the elements in $E_v$ and $L_v$ are closed-twins of $v$. Observe that these elements have order $p^r$ for some $r\leq i$. Next, we show that all closed-twins of $v$ in $U_v$ have orders of the form $p^l$ for some $l>i$. Suppose not, then let $u$ be a closed-twin of $v$ in $U_v$. As $u \in U_v, \ p^i $ divides $ o(u)$. Then $o(u)=k.p^i$, where $k>1$ and $gcd(k,p)=1$. Since $u$ is a closed-twin of $v$, $|\Gamma_v|-1=deg_{\Gamma_v}(u)=deg_{\Gamma_v}(v)=deg_{\Gamma}(v)$. Now, $\langle u \rangle$ has an element of order $k$ and this element cannot be a neighbor of $v$. So, $deg_{\Gamma}(u) > deg_{\Gamma_v}(u)= deg_{\Gamma}(v)$. Therefore, $deg_{\Gamma}(u) > deg_{\Gamma}(v)$. It is a contradiction. Hence, $o(u)=p^l, \ \text{for some } l>i$.

Given that $y$ is the closed-twin of $v$ in $\Gamma_v$ with  maximum order, say $p^j$. Suppose $z\in \langle y \rangle $. If $y\in E_v\cup L_v$, then clearly $\langle y \rangle =\langle v \rangle $ (because $y$ cannot be in $L_v$). If $y\in U_v$, then noting that $deg_{\Gamma}(y)\leq deg_{\Gamma}(v)$ and $y$ is a closed-twin of $v$ in $\Gamma_v$, we can say that $z$ is in $\Gamma_v$. In both the cases $\langle y \rangle \subseteq V(\Gamma_v)$. We show that every vertex $w\in V(\Gamma_v)$ is adjacent to $z$. Since $w\in V(\Gamma_v)$ and $y$ is a closed-twin of $v$, there is an edge between $w$ and $y$. So, $o(y)|o(w)$ or $o(w)|o(y)$. In the first case, $z\in \langle y \rangle \subseteq \langle w \rangle $. On the other hand, if $o(w)|o(y)$, then $w\in \langle y \rangle $. So either $z\in \langle w \rangle $ or $w\in \langle z \rangle $ as $\langle y \rangle $ is a cyclic group of prime power order. In any case, $\{w,z\}$ is an edge. So, any element $z$ in $\langle y \rangle $ is a closed-twin of $v$.
 
Let $z$ be a closed-twin of $v$. If $z\in \langle v \rangle $ then $z\in \langle y \rangle $. On the other hand, if $z\notin \langle v \rangle $, then $z\in U_v$. Therefore, $o(z)$ is a power of $p$. As $y$ is a closed-twin of $v$, there is an edge between $y$ and $z$. Therefore, $z\in \langle y \rangle $ or $y\in \langle z \rangle $. If $y\in \langle z \rangle $, we must have $\langle y \rangle =\langle z \rangle $ as $o(z)\leq o(y)$ (as both are $p$-power order closed-twins of $v$ and $y$ is with maximum order). This forces $\langle z \rangle =\langle y \rangle $. Thus, $z\in \langle y \rangle $ in both cases. This completes the proof of part (1). 

Now, we prove part (2). Let $x\in V(\Gamma_v)\setminus \langle y \rangle $. In this case, $x\notin E_v \cup L_v$ by \Cref{lemma 1a}. Since $y$ is a closed-twin of $v$, $ \{ x,y \}$ is an edge. Therefore, $x\in \langle y \rangle $ or $y\in \langle x \rangle $. However, by assumption, $x\notin \langle y \rangle $. So, $y\in \langle x \rangle $. Therefore, $o(y)|o(x)$. So, $o(x)=p^j\cdot k$ for some $k>1$.

Therefore, $ V(\Gamma_v)=\langle y \rangle  \sqcup \ \{x\in V(\Gamma_v) : p^j|o(x)\text{ and } o(x)\neq p^j\}$, i.e., $V(\Gamma_v)=\langle y \rangle \sqcup S$. 
 
To prove part (3), we define an equivalence relation $\equiv$ on $S$ as follows: $x_1 \equiv x_2,$ if and only if $\langle x_1 \rangle =\langle x_2 \rangle $. Note that the generators of $\langle x_1 \rangle $ are in $S$ by \Cref{generators and twins}. Therefore, the equivalence class of any vertex $x \in S$ is of size $\phi(o(x))$. Recall that, $o(x)=o(y)\cdot k=p^j\cdot k$. Now, as $p^j \geq p^2$, so $p$ divides $ \phi(o(x))$. Therefore, $p$ divides $ |V(\Gamma_v)|$, as claimed.
\end{proof}

\section{Finding a CCG-set of a group from its power graph and enhanced power graph}\label{ccg from a graph}

Given a directed power graph, the set of vertices corresponding to a CCG-set $\{g_1,\dots,g_m\}$ of the underlying group $G$ can be readily found in the graph. The scenario changes when the input graph is a power graph or an enhanced power graph and the underlying group is not given directly. Then, it is not possible to recognise these vertices exactly in the input graph as we can not distinguish two closed-twins $g_i$ and $g_i'$ in $Pow(G)$ (or $EPow(G)$). For example, if we take $\mathbb{Z}_{p^m}$ for some prime $p$ and integer $m$, then $Pow(\mathbb{Z}_{p^m})$ is a clique [\Cref{complete power graph}]. If the vertices of $Pow(\mathbb{Z}_{p^m})$ are named arbitrarily, then it is not possible to distinguish a generator of $\mathbb{Z}_{p^m}$ from any other vertex. Fortunately, the fact that the underlying group is $\mathbb{Z}_{p^m}$ can be concluded just from the graph by \Cref{complete power graph}. 

Therefore, we aim to do the following: Given a power graph (or an enhanced power graph) $\Gamma$, mark a set $\{g_1,g_2,\dots,g_m\}$ of vertices such that (1) each $g_i$ is a CC-generator or $g_i$ is a closed-twin of a CC-generator $g'_i$ in the graph $\Gamma$, and (2) $\{h_1,h_2,\dots,h_m\}$ is a CCG-set where $h_i=g_i$, if $g_i$ is a CC-generator; otherwise, $h_i=g_i'$. 


In \Cref{section: CCG of power graph}, we find the vertices corresponding to a CCG-set of the underlying group of the power graph.
The process of finding a CCG-set for the underlying group of an enhanced power graph is discussed in \Cref{ccg to epow}.

\subsection{Finding a CCG-set of a group from its power graph}\label{section: CCG of power graph}
\begin{theorem}\label{ccg from pow}
    There is an efficient polynomial time algorithm that, on input a power graph \footnote{Recall that the underlying group is not given.} $\Gamma \in \mathscr{P}ow$, outputs a set $\{g_1,g_2,\dots,g_m\}$ where $g_i$ is a CC-generator or $g_i$ is a closed-twin of a CC-generator $g'_i$ in the graph $\Gamma$ such that $\{h_1,h_2,\dots,h_m\}$ is a CCG-set where $h_i=g_i$, if $g_i$ is a CC-generator, otherwise, $h_i=g_i'$.
\end{theorem}

Hence, without loss of generality, we call the set $\{g_1,g_2,\ldots, g_m\}$ as CCG-set and $g_i$'s as CC-generators. 

Before we give the proof of the above theorem, we need the following definition that is required in the algorithm.

\begin{definition}
 Let $d$ be a positive integer.  Let $\Gamma \in \mathscr{P}ow$ and $v$ be a vertex in  $\Gamma$.
    Let $T_1,T_2,\ldots,T_r$ be the closed-twin partition
    of vertices in $\Gamma_v$. Let $S_1,S_2,\ldots, S_{r'}$ be the closed-twin partition of $Pow(\mathbb{Z}_d)$. We say that $\Gamma_v$ \emph{closed-twin-partition-wise matches} with $Pow(\mathbb{Z}_d)$ if (1) the closed-twin class containing the dominating vertices of both the graphs have the same size, and (2) $r=r'$ and there is some permutation $\pi\in Sym(r)$ such that $|T_i|=|S_{\pi(i)}|$.
\end{definition}

If $v$ is a CC-generator and $o(v)=d$ is not a prime power. Then, $\Gamma_v$ twin-partition-wise-matches with $Pow(\mathbb{Z}_d)$, by \Cref{corollary twin size}.

It is not hard to see that testing if $\Gamma_v$ closed-twin-partition-wise matches with $Pow(\mathbb{Z}_d)$ can be done in polynomial time. Also, when $d$ is not prime power, the size of the closed-twin class containing $v$ has size $\phi(d)$+1.

\begin{proof}\label{Correctness of algorithm of to find cyclic cover in given power graph}
(Proof of \Cref{ccg from pow})    The process of finding a CCG-set of the underlying group of a given power graph is described in \Cref{algo to find ccg from a power graph}.
 
\begin{algorithm}[ht]
\caption{Algorithm to mark a CCG-set in a finite power graph}\label{algo to find ccg from a power graph}
\hspace{0.5cm}\textbf{Input:} $\Gamma \in \mathscr{P}ow$
\vspace{0.2cm}
\begin{itemize}
    \item First, isolate the case when the power graph $\Gamma$ is a clique using \Cref{complete power graph}. Then return a singleton set, consisting of any vertex, as the CCG-set.
    \item If $\Gamma$ is not a clique, then mark any of the universal vertices as the identity.
    \item Next, all vertices except the identity are stored in a list $L$ in decreasing order of their degrees.
    \item During the algorithm, we use the labels: U (undecided), CC (a CC-generator) and NC (not a CC-generator). 
    \item To start with, mark all the vertices U in the list. Note that identity is not marked with any label.
    \item The algorithm marks the vertices further in phases. In each phase, pick the first U marked vertex, say $v$, in the list $L$ and do the following:

\textbf{[Rule 1a]} \textit{If} $deg(v)+1$ is a prime power and $\Gamma_v=\Gamma[N[v]]$ is complete, then mark $v$ as CC and mark all its neighbors NC. 

\textbf{[Rule 1b]} \textit{Else if} $deg(v)+1$ is a prime power and $\Gamma_v$ is not complete, then mark $v$ as NC. 
\textbf{[Rule 2a]}
\textit{Else if} $deg(v)+1$ is not a prime power and if $v$ has a closed-twin $w$ in $\Gamma_v$ such that $w$ has been marked NC, then mark $v$ as NC. 

\textbf{[Rule 2b]}
\textit{Else} (i.e., $deg(v)+1$ is not a prime power and $v$ does not have a NC marked closed-twin in $\Gamma_v$)

\hspace{1.5cm} \textit{If} $\Gamma_v$ closed-twin-partition-wise matches with $Pow(\mathbb{Z}_d)$, where $d=deg(v)+1$

\hspace{2cm} Mark $v$ as CC and all its neighbors NC.

\hspace{1.5cm}\textit{Else}

\hspace{2cm}Mark $v$ as NC.
\item Return the set of vertices marked CC.
\end{itemize}
    
\end{algorithm}

\Cref{algo to find ccg from a power graph} terminates after finitely many steps since vertex picked at each phase is either marked CC or NC thus reducing the number of vertices marked U in each phase.

Now we prove that the above algorithm is correct. The proof is by induction on the number of phases. In any phase, one of the four rules is applied to relabel a set of vertices. Our goal is to prove that this labelling is done correctly. In phase $i$, we assume that up to phase $(i-1)$, all the labellings were done correctly. For the base case, this means that all the vertices are still labelled U.

\textit{If Rule 1a is applied:} 
If $v$ is not a CC-generator, then $v$ is contained in at least one covering cycle. If $v$ is contained in two covering cycles, say $\langle g_1 \rangle$ and $\langle g_2 \rangle$, then $\Gamma_v$ is not complete, as the CC-generators $g_1$ and $g_2$ are not adjacent to each other. Now consider the case when $v$ is contained in exactly one covering cycle, say $\langle x \rangle$. Then $N_{\Gamma_v}(v)\subseteq N_{\Gamma_v}(x)$. So, if $deg_{\Gamma}(x)> deg_{\Gamma}(v)$, then $x$ or one of its closed-twins has already been marked as CC in some previous phase, and then $v$ would have been marked as NC. Now if $deg_{\Gamma}(x) = deg_{\Gamma}(v)$, then $v$ and $x$ are closed-twins and therefore $v$ is also a CC-generator. This is a contradiction.

\textit{If Rule 1b is applied:} If $v$ is a CC-generator, then $\Gamma_v$ is a complete graph. Thus, this step works correctly.

\textit{If Rule 2a is applied:} If $v$ is a CC-generator, then by \Cref{generators are the only twins} its closed-twins in $\Gamma_v$ are exactly $e$ (identity) and generators of $\langle v \rangle$. So, if any of the closed-twins is marked NC, it must have been because some other closed-twin is already marked CC in some previous phase $t \leq i-1$ of the algorithm. In phase $t$, the algorithm would have also marked $v$ as NC.  

\textit{If Rule 2b is applied:} If $v$ is a CC-generator, then $\Gamma_v$ closed-twin-partition-wise matches with $Pow(\mathbb{Z}_d)$. Now if none of $v$'s closed-twins in $\Gamma_v$ are already marked CC, then $v$ can be marked CC.

On the other hand, suppose that $v$ is not a CC-generator. We first consider the case when $v$ is contained in only one covering cycle, say generated by $x$. 

\begin{claim}
    $deg_\Gamma(x) > deg_{\Gamma}(v)$. 
\end{claim}
\begin{proof}
    As $v$ is contained in only one covering cycle, we have $N_{\Gamma}(v)\subseteq N_{\Gamma}(x)$. This implies $deg_{\Gamma}(x)\geq deg_{\Gamma}(v)$. Moreover in $\Gamma_v$, the vertices $x$ and $v$ are closed-twins. If $o(x)=p^i,$ then $deg(x)+1=p^i$. The graph $\Gamma_x = \Gamma[N[x]]$ is complete. So, $deg(v)+1=p^i$. Therefore, this case cannot arise. On the other hand, if $o(x)$ is not a prime power, we can apply \footnote{Here $x$ and $v$ are to be treated as the variables $v$ and $u$ in  \Cref{generators are the only twins}.} \Cref{generators are the only twins} and since $v\neq e$ and $v$ is not a CC-generator, we can see that $v$ and $x$ are not closed-twins in $\Gamma_x$. Hence, we have $deg_{\Gamma}(x)>deg_{\Gamma}(v)$.
\end{proof}
By the above claim, the algorithm considers $x$ and other generators of $\langle x \rangle$ before $v$. Then, by the induction hypothesis, one of these generators would be marked CC, and $v$ would not be labelled U.

Now we consider the case when $v$ is contained in at least two covering cycles, say $\langle g_1 \rangle$ and $\langle g_2 \rangle$. We prove that if $v$ is not a CC-generator, then $\Gamma_v$ cannot closed-twin-partition-wise match with $Pow(\mathbb{Z}_d)$. This case is divided into two subcases.

In the $1^{st}$ subcase, we assume that $o(v)$ is not a prime power. Now we count the closed-twins of $v$ in 
$\Gamma_v$ present in each of the sets $U_v,\ E_v$ and $L_v$.

If $x\in U_v$ is a closed-twin of $v$ in $\Gamma_v$, then by \Cref{twin in U has more degree}, $deg_{\Gamma}(x) > deg_{\Gamma}(v)$. So, $x$ must have been considered by the algorithm before $v$. At that phase, the algorithm either marked $x$ as NC or CC. If $x$ was marked as NC, $v$ would not satisfy the condition of Rule 2b (i.e., no closed-twin of $v$ in $\Gamma_v$ is marked NC). On the other hand, if $x$ was marked CC, the algorithm would have marked $v$ as NC. So, there are no closed-twins of $v$ in $U_v$.

Number of closed-twins of $v$ in $\Gamma_v$ which are present in $E_v$ is $\phi(o(v))$. By noting that $\Gamma_v[E_v\sqcup L_v]=Pow(\langle v \rangle)$ and using \Cref{generators are the only twins} on $Pow(\langle v \rangle)$, we see that the only closed-twin of $v$ in $L_v$ is the identity. Therefore, the total number of closed-twins of $v$ in $\Gamma_v$ is $\phi(o(v))+1$.

Now CC-generators $g_1$ and $g_2$ have distinct \footnote{$g_1$ and $g_2$ are not adjacent.} closed-twin-classes of size at least $\phi(o(g_1))$ and $\phi(o(g_2))$. But, $\phi(o(g_i))\geq \phi(o(v))$ for $i=1,2$ by \Cref{phi} . This is a contradiction since $Pow(\mathbb{Z}_d)$ can have at most two closed-twin-classes of size greater than or equal to $\phi(o(v))$, by (3) of \Cref{corollary twin size}. 

In the $2^{nd}$ subcase, we assume that $o(v)$ is a prime power, say $o(v)=p^i$ for some prime $p$ and some integer $i>0$. Consider $y$ and $S$ as in \Cref{v prime power order}. Note that $deg_{\Gamma}(y) \leq deg_{\Gamma}(v)$. Since otherwise, the algorithm would have marked $y$ as NC or CC. In both cases, the algorithm would not satisfy the conditions of Rule 2b.

\textit{Subsubcase 1:} $o(y)=p^j, j\geq 2$. 
In this case, by (3) of \Cref{v prime power order}, $p$ divides $ \lvert V(\Gamma_v) \rvert$. Therefore, $\Gamma_v$ must have a closed-twin-class of size $p-1$ for it to closed-twin-partition-wise match with $Pow(\mathbb{Z}_d)$ (because of (2) of \Cref{corollary twin size}). By (1) of \Cref{v prime power order}, if $x \in \langle y \rangle$, then the number of closed-twins of $x$ in $\Gamma_v$ is $|\langle y \rangle|=p^j > p-1$. 
Also, if $x\in S$, then number of closed-twins of $x$ in $\Gamma_v$ $\geq \phi(o(x)) \geq \phi(o(y)) \geq \phi(p^2) > p-1$. Therefore, there is no closed-twin-class of size $p-1$. 

\textit{Subsubcase 2:} $o(y)=p$. 
Recall that $\langle g_1 \rangle$ and $\langle g_2 \rangle$ are covering cycles containing $v$. Since $y$ and $v$ are closed-twins in $\Gamma_v$, we can see that $y \in \langle g_1 \rangle \cap \langle g_2 \rangle$. Now by (1) of \Cref{v prime power order}, the size of the closed-twin-class of $v$ is $p$. Since $o(y) \big| o(g_1)$ and $o(y) \big| o(g_2)$, the size of the closed-twin-class of both $g_1$ and $g_2$ is at least $p-1$. This is not possible by (3) of \Cref{corollary twin size}.
\end{proof}


\subsection{Finding a CCG-set of a group from its enhanced power graph}\label{ccg to epow} 
\begin{lemma}
\label{neighborhood in enhanced power graph}
    If $v$ is a CC-generator of a group $G$, then $N_{EPow(G)}[v] \subseteq N_{EPow(G)}[u]$ for all $u \in \langle v \rangle $.
\end{lemma}
\begin{proof}\label{Proof:neighborhood in enhanced power graph}

    Since $v$ is a $CC$-generator, $N_{EPow(G)}[v]\text{ corresponds to }\langle v \rangle$. Let $x \in N_{EPow(G)}[v]$. Then, $x \in \langle v \rangle$. Now if $u \in \langle v \rangle$, then by the definition of enhanced power graph, $x$ and $u$ are adjacent in $EPow(G)$. Hence, $x \in N_{EPow(G)}[u]$.
\end{proof}
\begin{theorem}
     There is an efficient polynomial time algorithm that on input an enhanced power graph\footnote{Recall that the underlying group is not given.} $\Gamma \in \mathscr{EP}ow$ outputs a set $\{g_1,g_2,\dots,g_m\}$ where $g_i$ is a CC-generator or $g_i$ is a closed-twin of a CC-generator $g'_i$ in the graph $\Gamma$ such that $\{h_1,h_2,\dots,h_m\}$ is a CCG-set where $h_i=g_i$ if $g_i$ is a CC-generator, otherwise, $h_i=g_i'$. 
\end{theorem}\label{CCG EPOW}

As before, we call the set $\{g_1,g_2,\ldots,g_m\}$ as CCG-set and $g_i$'s as CC-generators. 

\begin{proof}
We describe an algorithm for the task. 
\begin{algorithm}
    \caption{Algorithm to mark a CCG-set of $G$ in a finite enhanced power graph}
    \label{$EPow$ to CCG}
    \hspace{0.2cm}\textbf{Input:} $\Gamma \in \mathscr{EP}ow$\\
    \begin{enumerate}
        \item Sort the vertices of $\Gamma$ by their degree in increasing order. Let the sorted array be $A=\{v_1,v_2,\dots,v_m\}$.
        \item Pick the first unmarked element $x$ of $A$ and mark it CC.\\
        Mark all the elements of $N[x]$ as NC.
        \item Pick the next unmarked element in $A$ and repeat Step 2 till all elements of $A$ are marked.
    \end{enumerate}
\end{algorithm}


The proof of correctness of \Cref{$EPow$ to CCG} is by induction on the number of iterations.
In any iteration, the first unmarked vertex is marked as CC and its neighbors in the graph are marked as NC. Our goal is to prove that this marking process is correct. 

For the base case, $x=v_1$. By \Cref{neighborhood in enhanced power graph}, $v_1$ is either a CC-generator or $v_1 \in \langle g_1 \rangle $, where $g_1$ is a CC-generator and $v_1$ is a closed-twin of $g_1$ in $\Gamma$. Since $N[v_1]$ \text{ corresponds to } $\langle g_1 \rangle$ by \Cref{neighborhood in enhanced power graph}, we can safely mark the vertices adjacent to $v_1$ as NC.

In phase $i$, we assume that up to iteration $(i-1)$, all the markings were done correctly. Let us pick the first unmarked vertex, say $x$, in $A$. It is easy to see that $x$ does not belong to any covering cycle marked till the $(i-1)^{th}$ iteration, i.e., $x$ does not belong to the neighborhood of any CC marked vertex till the $(i-1)^{th}$ iteration. So, again using the same argument given in the base case, it can be seen that the markings done in the $i^{th}$ iteration are correct.
\end{proof}


\label{Correctness of the algorithm ccg of epow}
\section{Isomorphism of directed power graphs}\label{section 6 ISO}


The isomorphism problems of power graphs, directed power graphs, and enhanced power graphs are equivalent (see \cite{cameron2021graphs,cameron2010power,zahirovic2020study}). Thus, an algorithm for the isomorphism problem of directed power graphs automatically gives an isomorphism algorithm for power graphs (or enhanced power graphs), provided we can obtain the directed power graph from the power graph (respectively, the enhanced power graph). This is done in \Cref{section 7 Q2}. In the currect section, we focus on the isomorphism problem of directed power graphs. In the last part of this section, we discuss a necessary result that is used in \Cref{section 7 Q2} for obtaining the directed power graph of an input power graph (or an enhanced power graph). 

We perform several reductions on a directed power graph that are isomorphism invariant.
 The out-degree of a vertex in $DPow(G)$ is the order of the element in the group $G$, i.e., for a vertex $u$, $out$-$deg(u)=o(u)$.  Therefore we can color the vertices by their out-degrees. We call the colored graph $CDPow(G)$. We emphasise that here the colors are numbers, and hence we can perform arithmetic operations on these colors and use the natural ordering of integers inherited by these colors. 
We recall that by isomorphism we mean color preserving isomorphism when the graphs  are be colored.

Two vertices $u$ and $v$ are closed-twins in $CDPow(G)$ (in $DPow(G)$ also) if and only if $\langle u \rangle =\langle v \rangle $ in $G$, i.e., $u$ and $v$ are two generators of the same cyclic subgroup in $G$. There are $\phi(o(u))$ generators of $\langle u \rangle $ in $G$. So, for each vertex $u \in CDPow(G),$ there are exactly $\phi(col(u))$ closed-twins in $CDPow(G)$.  By the converse of Lagrange's theorem, in each cyclic subgroup of order $n$, for each divisor $k$ of $n$, there are exactly $\phi(k)$ generators. So, for each $k$ in the color set of $CDPow(G)$, there are $\phi(k)$ closed-twins in the graph. Observe that $u$ and $v$ are closed-twins in $CDPow(G),$ if and only if $(u,v) \in E(CDPow(G))$ and $col(u)=col(v)$.

\textbf{\emph{Reduction rule 1:}} \textbf{Closed-twin Reduction}: If there are two closed-twins $u$ and $v$ in $CDPow(G)$, then do a vertex identification of $u$ and $v$ and color the identified vertex with $col(u)=col(v)$. Let $R_1(G)$ denotes the reduced graph after applying Reduction rule 1 to $CDPow(G)$.

 From the discussion above, the next lemma follows easily.

\begin{lemma}
\label{twin reduction is safe}
$CDPow(G) \cong CDPow(H)$ if and only if $R_1(G) \cong R_1(H)$. 
\end{lemma}


    \begin{remark}
    It is easy to see that, we can get back an isomorphic copy of $CDPow(G)$ from $R_1(G)$, by adding $\phi(col(u))$ closed-twins at each vertex $u$ in $R_1(G)$.
\end{remark}

    Since we know that each vertex has a self-loop, for the purpose of isomorphism we can delete these self-loops. It is also easy to check that $R_1(G)$ is a transitively closed directed graph. 
    

\textbf{\emph{Reduction rule 2: }}\textbf{Edge-deletion:}
Let us consider $R_1(G)$. Do the following steps:\\
(1) Delete all self-loops.\\
(2) For all $a,b,c$, if $(a,b)$ and $ (b,c) $ are edges, then mark $(a,c)$ as a transitive edge. Then, delete all edges that are marked as transitive edges.
Let $R_2(G)$ denotes the resulting graph.

Since $R_1(G)$ is the reflexive and transitive closure of $R_2(G)$, we have the following lemma: 
\begin{lemma}
\label{edge deletion safe}
$R_1(G) \cong R_1(H)$ if and only if $R_2(G) \cong R_2(H)$.
\end{lemma}

\begin{lemma}   
\label{structure of R_2}
The reduced graph $R_2(G)$ satisfies the following properties:
\begin{enumerate}
    \item Vertices with in-degree zero in $R_2(G)$ form a CCG-set of $G$.
    \item If $(u,v)$ is an edge in $R_2(G)$, then $col(u) > col(v)$. Moreover, $col(u)=col(v)\cdot p$ for some prime $p$.
    \item $R_2(G)$ is a directed acyclic graph.
\end{enumerate}

\end{lemma}
\begin{proof}
\begin{enumerate}
    \item \label{Proof: generator_in-degree}
Let us assume that $u\in V(R_2(G))$ is a vertex of in-degree zero, i.e., there is no incoming edge $(v,u)$ to $u$ for any $v\in V(R_2(G))$. This implies that in $DPow(G)$, the only incoming edges to $u$ are from its closed-twins. So, $u \notin \langle v \rangle $ for any $v$. Thus, $\langle u \rangle $ is a maximal covering cycle. Hence $u$ is a covering cycle generator (CC-generator).

 \item \label{Proof:degree_claim}
Since $(u,v)$ is an edge in $R_2(G)$, $v$ is generated by $u$ in $G$. So, $col(v)\mid col(u)$ and hence $ col(u)\geq col(v)$. Now, if $ col(u) = col(v)$, then $ \langle u \rangle =\langle v \rangle $. This means that $u$ and $v$ are closed-twins, which must have taken part in the vertex identification process in Reduction rule 1 itself. So, we can discard this case and the only possibility we are left with is $col(u)>col(v)$.

Since $\langle v \rangle \subsetneq \langle u \rangle$, we have $\langle u \rangle /\langle v \rangle \cong \mathbb{Z}_m$ for some $m$. If $m$ is not a prime, then $\mathbb{Z}_m$ has a proper nontrivial subgroup $H$ in it. So, $\langle u \rangle/ \langle v \rangle$ also has a subgroup $\langle w \rangle/ \langle v \rangle$ which is isomorphic to $H$. Therefore, $\langle v \rangle \subsetneq \langle w \rangle \subsetneq \langle u \rangle$. Now, in $R_2(G)$ there is a closed-twin of $w$, say $w'$. This means $(u,w')$ and $(w',v)$ are two edges in $R_2(G)$. But in that case, we would have marked $(u,v)$ as a transitive edge during the reduction from $R_1(G)$ to $R_2(G)$ and hence $(u,v) \notin E(R_2(G))$. This is a contradiction. So, $m$ is some prime $p$ and $|\langle u \rangle |/|\langle v \rangle|=p$. Hence, $col(u)=col(v)\cdot p$.

\item \label{Proof:DAG}
We prove this by contradiction. Let us assume the graph has a cycle, say $a_1a_2...a_ka_1$. Then, by (2) of ~\Cref{structure of R_2}, we have $col(a_1)>col(a_2)>\dots >col(a_k)>col(a_1)$, which is not possible. So, our assumption is wrong. That means $R_2(G)$ is acyclic.\end{enumerate}\end{proof}

Note that using (1) of \Cref{structure of R_2}, we can easily find a set of vertices, say $\{g_1,g_2,\ldots,g_m\}$, that form a covering cycle generating set (CCG-set) of $G$. 

\textbf{\emph{Reduction rule 3:}} \textbf{Removing the direction}: Remove the direction of the edges in $R_2(G)$ to obtain an undirected colored graph $R_3(G)$.

 Note that the CCG-set of $G$ can still be identified easily in $R_3(G)$: A vertex $g$ is a CC-generator if and only if all its neighbours have smaller orders (or colors).
 
 The following result is an easy consequence of (2) of \Cref{structure of R_2}. 

\begin{lemma}
\label{removing direction is safe}
 $R_2(G) \cong R_2(H)$ if and only if $R_3(G) \cong R_3(H)$.
\end{lemma}
\begin{definition}
    A path $u_1u_2\ldots u_l$ in $R_3(G)$ is said to be a \emph{descendant path} if $col(u_i) > col(u_{i+1})$. The vertices in the graph reachable from $u$ using descendant path are called \textit{descendant reachable} vertices from $u$. We denote the set of descendant reachable vertices from $u$ in $R_3(G)$ by $Des(u)$.
\end{definition}
 Observe that $Des(u)$ in $R_3(G)$ is same as the closed out-neighborhood of $u$ in $R_1(G)$. The colors of the vertices of $Des(u)$ in $R_3(G)$ form the set of all divisors of $col(u)$. Also, no two vertices of $Des(u)$ in $R_3$ have the same color.

\begin{theorem}\label{R_3 of p-group is tree}
 If $G$ is a finite p-group, then $R_3(G)$ is a colored tree.
\end{theorem}
\begin{proof}

Let $|G|=p^\alpha$. Any edge in $R_3(G)$ is of the form $\{u,v\}$ where by using (2) of \Cref{structure of R_2} we can assume without loss of generality that $col(u)=p^t$ and $col(v)= p^{t-1},$ for some $t \in \{1,2,\dots,\alpha\}$. Suppose the graph contains a cycle $ u_0 u_1 u_2 \dots u_n u_0 $. By (2) of \Cref{structure of R_2}, we can assume without loss of generality that the colors of the vertices form the following sequence: $p^tp^{t-1}p^{t-2} \dots p^{t-(i-1)}p^{t-i}p^{t-(i-1)} \dots p^{t-1}p^{t}$ for some $i$. Now, $u_1, u_n \in Des(u_0)$ such that $col(u_1)=col(u_n)=p^{t-1}$. This is a contradiction since no two vertices of $Des(u)$ for any vertex $u$ have the same color. Hence, our assumption is wrong and $R_3(G)$ has no cycle. 
\end{proof}

Since the isomorphism of trees can be tested in linear time (see, for example, \cite{aho1974design}), the isomorphism of the directed power graphs arising from $p$-groups can also be tested in linear time. 

 
Now we extend our algorithm to check the isomorphism of directed power graphs of finite nilpotent groups. 
For that, we use the following two results.

\begin{lemma}\cite{mukherjee2019hamiltonian}
\label{directed power graph of direct product is strong product}
    Let $G_1$ and $G_2$ be two finite groups such that $|G_1|$ and $|G_2|$ are co-prime to each other. Then, $DPow(G_1\times G_2)=DPow(G_1)\boxtimes DPow(G_2)$, where $\boxtimes$ denotes the strong product of two graphs.
\end{lemma}

\begin{lemma}\cite{mckenzie1971cardinal}
\label{unique prime factorisation of strong product of digraphs}
There exists a unique prime factor decomposition of a  simple connected \footnote{Here \textit{connected directed graph} means that the underlying undirected graph is connected.} directed graph with respect to strong product and the uniqueness is up
to isomorphism and ordering   of the factors.
\end{lemma}
One can easily verify that the following lemma holds from the above two lemmas. However, we can prove  \Cref{isomorphism of strong product} without using \Cref{unique prime factorisation of strong product of digraphs} and the  proof is given in \Cref{proof of lemma 35}.



\begin{restatable}{lemma}{isoStrongProduct}
\label{isomorphism of strong product}
Let  $G=G_1\times G_2 \times \dots \times G_k$ and   $H=H_1\times H_2 \times \dots \times H_k$ where $|G_i|= |H_i|$ 
 for all $1 \leq i \leq k$. Suppose that $gcd(|G_i|,|G_j|)=gcd(|H_i|,|H_j|)=1$, for all $1\leq i < j \leq k$. Then, $DPow(G)\cong DPow(H)$ if and only if $DPow(G_i)\cong DPow(H_i)$, for all $1 \leq i \leq k$.
    
\end{restatable}

We are now ready to present one of the main results of the paper. Namely, we show that the isomorphism of the directed power graphs of nilpotent groups can be tested in polynomial time. Let $\mathcal{DP}ow_{nil}=\{DPow(G)\ : \ G \textrm{ is a finite nilpotent group}\}$. 

Theorem \ref{R_3 of p-group is tree} and  Lemma \ref{isomorphism of strong product} suggest that obtaining the directed power graphs corresponding to the factor groups might be useful. One approach would be to decompose an input directed power graph into prime factors with respect to strong product in polynomial time using the algorithm by Hellmuth et al. \cite{hellmuth2015cartesian}. Note that, in a general setting, the prime graphs in the strong product decomposition may not correspond to the directed power graphs of the direct factors of the underlying group. We are also not sure if the Sylow-$p$ subgroups of a nilpotent group generate prime graphs. If not, then just applying the algorithm of Hellmuth et al. is not enough and we need to regroup the prime factors properly to apply Theorem \ref{R_3 of p-group is tree}. Fortunately, all these complications can be easily avoided as shown in the next theorem. 
\begin{theorem}\label{Iso nilpotent}
    There is an efficient polynomial time algorithm that on inputs $\Gamma_1,\Gamma_2 \in \mathcal{DP}ow_{nil}$ checks if $\Gamma_1$ and $\Gamma_2$ are isomorphic.
\end{theorem}
\begin{proof}

We know that a finite nilpotent group is the direct product of its Sylow subgroups.  Since the orders of the Sylow subgroups are coprime with each other, by \Cref{isomorphism of strong product}, $\Gamma_1$ and $\Gamma_2$ are isomorphic if and only if for each prime $p$ dividing $|V(\Gamma_1)|$ (which is same as the order of the underlying group), the directed power graphs of the Sylow $p$-subgroups of the underlying groups of $\Gamma_1$ and $\Gamma_2$ are isomorphic.

Therefore, if we can find the directed power graphs of the Sylow subgroups associated with each prime divisor, we can test the isomorphism of  $\Gamma_1$ and $\Gamma_2$.

Note that, the underlying groups are not given as input. However, we can still compute the directed power graph of a Sylow $p$-subgroup of an input graph by finding the set $V_p$ of all vertices whose order in the underlying group is $p^i$ for some $i\geq 0$. More precisely, the subgraph induced by the set $V_p$ is the directed power graph associated with the Sylow $p$-subgroup. Note that the order of a vertex (which is also an element in the underlying group) is just the out-degree of the vertex in the directed power graph.   
\end{proof}

     
        

    We show that all the isomorphism invariant information of $R_3(G)$ is captured by a) the CCG-set of $G$ in $R_3(G)$ along with their colors, and b) elements corresponding to their pairwise common neighborhood along with their colors. For this, we do a further reduction. The results in the rest of this section is required in \Cref{section 7 Q2}.
    
We define a new simple undirected colored graph $HD[n]=(V,E)$ for any natural number $n$, where $V=\{ d : d|n\}$. The name of each vertex is treated as its color, i.e., here $col(v)=v$.
The edge set is $E=\{ \{u,v\}: v=u\cdot p \text{ or } u=v\cdot p \text{ for some prime $p$}\}$.
One can see that $HD[n]$ is the Hasse diagram of the POSET defined over the set of all divisors of $n$ with respect to the divisibility relation. Moreover, $HD[n]$ is also isomorphic to $R_3(\mathbb{Z}_n)$ (as a consequence of (2) of \Cref{structure of R_2}).

\begin{remark}
\label{Des(g_i)=HD[g_i]}
(1) It is easy to see that $R_3(G)[Des(g_i)]$ is isomorphic to $HD[col(g_i)]$ for all $1\leq i \leq m$. We can see that the isomorphism is unique as in each of these graphs, there is only one vertex with a particular color.\\
(2) Note that $\{y,y'\}\in E(R_3(G))$ if and only if (a) $y$, $y' \in Des(g_i)$ for some $1\leq i\leq m$ and (b) $col(y)=p\cdot col(y')$ or $col(y')=p\cdot col(y)$ for some prime $p$.
\end{remark}

Let $\Bar{I}(i,j)$ denotes the vertex in $R_3(G)$ that is of maximum color among the common descendant reachable vertices from both $g_i$ and $g_j$. It is not hard to see that in the group $G$, $col(\Bar{I}(i,j))=|\langle g_i \rangle \cap \langle g_j \rangle |$. Note that for two distinct pairs $(i,j)$ and $(i',j')$, $\Bar{I}(i,j)$ and $\Bar{I}(i',j')$ can be the same vertex in $R_3(G)$.

  \begin{claim}
 In $R_3(G)$, $gcd(col(\Bar{I}(i,j)),col(\Bar{I}(s,j)))$ divides $col(\Bar{I}(i,s))$.
 \label{intersection}
\end{claim}
\begin{proof}

Let $d=gcd(col(\Bar{I}(i,j)),col(\Bar{I}(s,j))$. Since $d$ is a factor of $col(\Bar{I}(i,j))$, there exists $y_1\in Des(\Bar{I}(i,j))=Des(g_i)\cap Des(g_j)$ such that $col(y_1)=d$. Similarly, there exists $y_2 \in Des(\Bar{I}(s,j))= Des(g_s)\cap Des(g_j)$ such that $col(y_2)=d$.

Since both $y_1$ and $y_2$ are descendants of $g_j$ with the same color, they must be the same element. Therefore, we can argue that $y_1$ is a common descendant of $g_i$ and $g_s$, i.e., $y_1\in Des(g_i)\cap Des(g_j)= Des(\Bar{I}(i,j))$. Hence, $col(y_1)\mid col(\Bar{I}(i,s))$.
\end{proof}

\textbf{\emph{Reduction rule 4}:}
Consider $R_3(G)$. Recall that, in $R_3(G)$ a CCG-set  $\{g_1, g_2,..., g_m\}$ of $G$ can be readily found. We make a new graph $R_4(G)$ as follows:\\
(1) Introduce the vertices $g_1,g_2,\dots,g_m$ with their colors.\\
(2) For each pair $(i,j),\ 1\leq i < j \leq m$, do the following:\\
     Find the vertex $\Bar{I}(i,j)$ that is of maximum color among the descendant reachable vertices from both $g_i$ and $g_j$. We add a vertex $I(i,j)$ in $R_4(G)$ and color it with $col(\Bar{I}(i,j))$. Add edges $\{g_i, I(i,j)\}$ and $\{g_j, I(i,j) \}$.


We can see that $R_4(G)$ is a bipartite graph where one part is a  CCG-set and another part contains vertices marked as $I(i,j)$ for all $(i,j)$. In $R_4(G)$, for distinct pairs $(i,j)$ and $(i',j')$, $I(i,j)$ and $I(i',j')$ are distinct vertices, while in $R_3(G)$, $\Bar{I}(i,j)$ and $\Bar{I}(i',j')$ may be the same vertex. In other words, $R_4(G)$ may have several copies of vertex $\Bar{I}(i,j)$.
We now present an algorithm to get back an isomorphic copy of $R_3(G)$ from $R_4(G)$.

 \textbf{Idea of the algorithm: }
    In $R_4(G)$, we have a set of colored CC-generators. Also, there exist vertices $I(i,j)$ corresponding to each pairwise intersection of maximal cyclic subgroups $\langle g_i \rangle $ and $\langle g_j \rangle $ in $G$. $I(i,j)$ is the only common neighbor of $g_i$ and $g_j$ in $R_4(G)$. Using this information, we construct $R_3(G)$ in an iterative manner. First, we describe a sketch of the idea behind the process. There are $m$ iterations in the process. In the $1^{st}$ iteration, we introduce $HD[col(g_1)]$. One can easily verify that $R_3(G)[Des(g_1)]$ is isomorphic to $HD[col(g_1)]$ (by (2) of \Cref{Des(g_i)=HD[g_i]}). In the $2^{nd}$ iteration, we introduce $HD[col(g_2)]$. As we know the color of $I(1,2)$, we have information about the set of vertices common to both $HD[col(g_1)]$ and $HD[col(g_2)]$. Let $u$ and $v$ be the vertices with color $col(I(1,2))$ in $HD[col(g_1)]$ and $HD[col(g_2)]$ respectively. We identify (via vertex-identification) the vertices with the same colors in $Des(u)$ (which is in $HD[col(g_1)]$) and $Des(v)$ (which is in $HD[col(g_2)]$)\footnote{Since $HD[col(g_i)]$ is isomorphic to $R_3(G)[Des(g_i)]$, we can use the concept of $Des$ in the graph $HD[col(g_i)]$ for all $i$.}. One can see that the resulting graph is isomorphic to the induced subgraph of $R_3(G)$ on $Des(g_1)\cup Des(g_2)$. 
  Inductively the algorithm introduces $Des(g_1)\cup Des(g_2) \cup \ldots \cup Des(g_{j-1})$ at the end of the $(j-1)^{th}$ iteration. In the $j^{th}$ iteration, we introduce $HD[col(g_j)]$. It is easy to note that the set of vertices in $HD[col(g_j)]$ that are contained in $Des(g_j)\cap Des(g_s)$ for all $s\leq j-1$ has already been introduced. So, we need to identify the vertices introduced by the algorithm earlier with the corresponding subset of vertices in $HD[col(g_j)]$.
    We get the information of such vertices using the color of $I(s,j)$ for $s\leq j-1$. The details of the algorithm is given below (\Cref{algorithm:R4 to R3}).


\begin{algorithm}[hbt!]
\caption{To construct an isomorphic copy of $R_3(G)$ from $R_4(G)$}\label{algorithm:R4 to R3}
\hspace*{\algorithmicindent} \textbf{Input:} $R_4(G)$ 

\begin{algorithmic}[1]

\State $X_1 \gets HD[col(g_1)]$
\State $j \gets 2$
\While{$j \leq m$}
    \State  Introduce $Y_j=HD[col(g_j)]$ 
    
    \State $s \gets 1$
    \State $h_{j,0} \gets \emptyset$ \Comment{Mapping for vertex identification}
    \While{$s \leq j-1$}
        \State
        Consider 
        $I(s,j)$.
       
        \State 
        \parbox[t]{\dimexpr\linewidth-\algorithmicindent}{$h_{j,s} \gets h_{j,s-1}\cup \{(u,v) : col(u)=col(v) \textrm{ where } u\in Des(g_s) \subseteq V(X_{j-1}) \textrm{ s.t } col(u)\big| col(I(s,j))\textrm{ and }$ $v \in V(Y_j)\}$} 
        \State $s \gets s+1$
    \EndWhile
    \State 
    \parbox[t]{\dimexpr\linewidth-\algorithmicindent} 
    {For all $(u,v) \in h_{j,j-1}$ do vertex identification of $u$ and $v$ and color the new vertex with $col(u)$.}
    \State $X_j \gets $ The new graph obtained after the above vertex identification process of $X_{j-1}$ and $Y_j$.
    
    \State $j \gets j+1$
\EndWhile
\State Return $X_m$ 
\end{algorithmic}
\vspace{0.5cm}

\end{algorithm}

As indicated above and in Line 12 of \Cref{algorithm:R4 to R3}, vertices in the old graph and $HD[col(g_j)]$ are identified. In \Cref{identification is conflict-free}, we show that these vertices can be identified without conflict.

\begin{restatable}{lemma}{goingbacklemma}\label{r4 to r3}
The graph $X_m$ returned by \Cref{algorithm:R4 to R3} is isomorphic to $R_3(G)$.
\end{restatable}
\begin{proof}\label{proof:r4 to r3}
The proof is technical and consists of some nuances. The details of the proof can be found in \Cref{Proof of Lemma 39}.



\end{proof}

\section{Reconstruction Algorithms}\label{section 7 Q2}
Cameron asked  the following question: ``Question 2  \cite{cameron2021graphs}: Is there a simple algorithm for constructing the directed power graph or enhanced power graph from the power graph, or the directed power graph from the enhanced power graph?'' Bubboloni and Pinzauti \cite{bubboloni2022critical} gave an algorithm to reconstruct the directed power graph from the power graph. In this section, we show that with the tools we have developed, we can readily design algorithms to reconstruct the directed power graph from both the enhanced power graph and the power graph. 

Suppose we are given a power graph (or an enhanced power graph) of some finite group $G$ as input, i.e., $\Gamma= Pow(G)$ (or, $\Gamma=EPow(G)$). However, the group $G$ is not given. As discussed in \Cref{ccg from a graph}, we can find a CCG-set for $G$ from the input graph. Next, we describe how to obtain a graph isomorphic to $R_4(G)$ from the CCG-set. From the vertices corresponding to a CCG-set of $G$, say $\{g_1,g_2,\ldots,g_m\}$, we get the information about their degree in $\Gamma$ and the pairwise common neighborhood of $g_i$ and $g_j$ in the respective graph. This immediately gives us $R_4(G)$. From $R_4(G)$, we know how to get back an isomorphic copy of $DPow(G)$ using the results in \Cref{section 6 ISO}. All the steps in the process can be performed in polynomial time. 

For any two vertices $u$ and $v$ we can easily decide when to put an edge between them in the enhanced power graph by looking into the corresponding directed power graph: there is an edge $\{u,v\}$ in the enhanced power graph, if and only if both $u$ and $v$ belong to the closed-out-neighbourhood of some
vertex in the directed power graph. In this way, we can construct the enhanced power graph from an input directed power graph. Therefore, we get a complete solution to Cameron's question.
\section{Proof of Lemma \ref{r4 to r3}}\label{Proof of Lemma 39}
In this section we prove \Cref{r4 to r3}.
\goingbacklemma*
\begin{proof}
    We show by induction on $j$ that the constructed graph up to the $j^{th}$ step is isomorphic to the subgraph of $R_3(G)$ induced on $Des(g_1)\cup Des(g_2) \cup\ldots \cup Des(g_j)$. This shows that after the $m^{th}$ iteration, we can get an isomorphic copy of $R_3(G)$.

\begin{claim}\label{X_j iso R3(j)}
  $X_j \cong R_3(G)[Des(g_1) \cup Des(g_2) \cup \dots \cup Des(g_j)], \ \forall 1\leq j \leq m $ .
\end{claim}
\textbf{Proof of claim:}
  For simplicity of writing, we denote $R_3(G)[Des(g_1) \cup Des(g_2) \cup \dots \cup Des(g_j)]$ by $R_3(j)$ in the remaining part of the proof. With this, $R_3(1)$ denotes $R_3(G)[Des(g_1)]$. 

  By \Cref{Des(g_i)=HD[g_i]}, $X_1= HD[col(g_1)]$ is isomorphic to $R_3(1)$ by a unique isomorphism, say $f_1$. If we take $f_0$ to be the empty map, then $f_1$ extends $f_0$.

We prove by induction on $j$ that $X_j$ is isomorphic to $R_3(j)=R_3[Des(g_1) \cup Des(g_2) \cup \dots \cup Des(g_j)]$ via a map $f_j$  that extends the isomorphism $f_{j-1}$.

    
  By induction hypothesis, let us assume that $X_{j-1}\cong R_3(G)[Des(g_1)\cup\dots\cup Des(g_{j-1})]$ and $f_{j-1}$ is an isomorphism between $X_{j-1}$ and $R_3(j-1)$ derived by extending $f_{j-2}$. We show that $f_j$ is an extension of $f_{j-1}$ and $f_j$ is an isomorphism between $X_j$ and $R_3(j)$. 

  However, before we go into the details of the inductive case, we address the following important issue.

  
  In the $j^{th}$ iteration of the outer while loop and just after the execution of Line 4 of \Cref{algorithm:R4 to R3}, the current graph is the disjoint union of  $X_{j-1}$ and $ Y_j$. Now to get $X_j$, some vertices of $X_{j-1}$ and $Y_j$ 
   are vertex-identified using the tuples stored in $h_{j,j-1}$ as described in Line 12 of \Cref{algorithm:R4 to R3}. Observe that two vertices in $Y_j$ cannot be identified with the same vertex in $X_{j-1}$,  because in $Y_j=HD[col(g_j)]$, no two vertices have the same color. However, there is a possibility that two or more vertices of $X_{j-1}$ are assigned to be identified with the same vertex of $Y_j$. We show that this case does not arise. To do this, we first define the following sets:
  \begin{center}
  $Y_{j,1}=\{v\in V(Y_j)\ : \ col(v)\big| col(I(1,j))\}$

$Y_{j,s}=Y_{j,s-1}\cup \{v\in V(Y_j)\ : \ col(v)\big| col(I(s,j))\},\ \ s=2,\dots,j-1$

$X_{j-1,1}=\{u\in V(X_{j-1}) \ : \ u\in Des(g_1) \textrm{ and }col(u)\big| col(I(1,j))\}$

$X_{j-1,s}= X_{j-1,s-1} \cup \{u\in V(X_{j-1})\ : \ u\in Des(g_s) \textrm{ and }col(u)\big| col(I(s,j))\}, \ \ s=2,\dots,j-1$
\end{center}

Now $h_{j,j-1}$ is updated from $h_{j,0}=\emptyset$ by the following rule:
$h_{j,s}=h_{j,s-1} \cup \{(u,v)\ | \ col(u)=col(v) \text{ where }  u\in X_{j-1,s} \textrm{ and } v\in Y_{j,s}\}$ (as described in Line 9 in \Cref{algorithm:R4 to R3})\footnote{Note that when $u\in X_{j-1,j-1}$ is identified with $v\in Y_{j,j-1}$, we color it with $col(u)$ and for simplicity we name the new vertex as $u$.}.  Since there is a unique vertex of any particular color in $Y_j$, we can see $h_{j,s}$ as a well-defined function from $X_{j-1,s}$ to $Y_{j,s}$. Now to show that $h_{j,j-1}$ gives a conflict-free vertex identification process, we show that $h_{j,s}$ is a bijection and an extension of $h_{j,s-1}$. Since $h_{j,s-1}\subseteq h_{j,s}$, it is enough to prove the following claim:
  
  \begin{claim}\label{identification is conflict-free}
     The map $h_{j,s}: X_{j-1,s}  \xrightarrow{} Y_{j,s}$ is a bijection, for all $ 1\leq s \leq j-1$. 
  \end{claim}
  \textbf{Proof of claim:}
    First, we show that $h_{j,s}$ is onto for all $s=1,\dots,j-1$. For this, take a vertex $v$ from $Y_{j,s}$. Then $col(v)|col(I(i,j))$ for some $i\leq s$. So,\footnote{Since $X_{j-1} \cong R_3(j-1)$, the concept of descendant reachability can also be defined in $X_{j-1}$. Therefore, it makes sense to use $Des(g)$ in $X_{j-1}$ for any vertex $g$.} there exists a vertex $u\in Des(g_i)$ in $X_{j-1,s}$ such that $col(u)=col(v)$ and $h_{j,s}(u)=v$.
    
    Now we prove that $h_{j,s}$ is one-to-one using induction on $s$. For the base case, it is easy to see that $h_{j,1}:X_{j-1,1} \xrightarrow{} Y_{j,1}$ is a bijection since $X_{j-1,1}$ and $Y_{j,1}$ contains colored vertices corresponding to each divisor of $col(I(1,j))$ and color of each vertex is distinct. By induction hypothesis we assume that $h_{j,s-1}: X_{j-1,s-1} \xrightarrow{} Y_{j,s-1}$ is a bijection. Now for the inductive case, we consider $h_{j,s}: X_{j-1,s} \xrightarrow{} Y_{j,s}$. We need to prove that $h_{j,s}$ is one-one. Suppose that $u \in X_{j-1,s}$ is paired with $v \in Y_{j,s}$ to be stored at $h_{j,s}$ in the $s^{th}$ iteration of the inner while loop (Line 9 of \Cref{algorithm:R4 to R3}). We need to argue that the pairing does not violate the one-to-one condition. We do this in two cases.
      
      Case 1: The vertex $v$ was not encountered in any of the previous iterations, i.e., $v \notin Y_{j,s-1}$.
      So by definition of $X_{j-1,s-1}$, there is no vertex of color $col(v)$ in $X_{j-1,s-1}$.
      Since $col(u)=col(v),$ we have $u\in X_{j-1,s}\setminus X_{j-1,s-1}$. So, $(u,v)$ is added to $h_{j,s}$ in the $s^{th}$ iteration only, where $v$ is in $Y_{j,s}$. Therefore, $X_{j-1,s}$ contains exactly one vertex of color $col(u)$. This implies that $v$ cannot be paired with any vertex except $u$.

      Case 2: The vertex $v$ was encountered before the $s^{th}$ iteration, and $i\leq (s-1)$ is the most recent such iteration. This means that there exists $u'$ in the old graph $(i.e., \  X_{j-1,s-1})$ such that $h_{j,s-1}(u')=v$. Since $h_{j,s-1}$ is a bijection by induction hypothesis, $u'$ is the only preimage of $v$ under $h_{j,s-1}$. We show that $u=u'$. 

      Observe that there is a vertex $w\in Des(g_s)$ in $X_{j-1}$ such that $col(w)=col(I(s,j))$. By the algorithm, $col(u)|col(I(s,j))$. So, $u \in Des(w)$.  Similarly, there is a vertex $w'\in Des(g_i)$ in $X_{j-1}$ such that $col(w')=col(I(i,j))$ and by the algorithm $ col(u')|col(I(i,j))$. So $u'\in Des(w')$. Since $col(u)=col(u')$, $col(u')|col(I(i,j))$ and $col(u)|col(I(s,j))$, we conclude that $col(u)$ divides $gcd(col(I(i,j)),col(I(s,j)))$. So, by Claim \ref{intersection}, $col(u)|col(I(i,s))$. 

    Now we consider the subgraph of $X_{j-1}$ induced by $Des(g_i)\cap Des(g_s)$. If $x \in Des(g_i)\cap Des(g_s)$ is the vertex with color $col(I(i,s))$, then this subgraph is formed by the descendants of $x$. Since the descendants of $x$ are exactly the vertices in $Des(g_i)$ and $Des(g_s)$ with colors as factors of $col(I(i,s))$, both $u$ and $u'$ are in $Des(x)$.
    Now, $Des(x)$ has a unique vertex of a particular color. Therefore, as $u$ and $u'$ have the same color, $u=u'$. 

    Hence, we have proved that $h_{j,s}$ is one-one in both the cases. Therefore, we can conclude that $h_{j,s}: X_{j-1,s}\rightarrow Y_{j,s}$ is a bijection for all $1 \leq s \leq j-1$. \hfill $\Box$
  
  From the above claim, we can conclude that in the $j^{th}$ iteration of the outer while loop, the identification process done in Line 12 in Algorithm \ref{algorithm:R4 to R3} via the mapping $h_{j,j-1}$ is correct. Next, we show that the graph $X_{j}$ (output in Line 13), derived after the identification process on $X_{j-1}$ and $Y_j$, is indeed isomorphic to $R_3(j)$. 

  For $j\geq 2$, we define $f_j:V(X_j)\xrightarrow{}V(R_3(j))$ in the following manner:
  \begin{equation}\label{defintion of f_j}
      f_j(x) = 
     \begin{cases}
        f_{j-1}(x) &\quad\text{if } x \in V(X_{j-1})\\
       y &\quad\text{otherwise}\\ \text{where } y \in V(R_3(j))\setminus V(R_3(j-1)) \text{ and } col(y)=col(x). \\ 
     \end{cases}
  \end{equation}

To show that $f_j$ is well defined, it is enough to argue that for each $x\in V(X_j)\setminus V(X_{j-1})$, there exists a unique $y\in V(R_3(j))\setminus V(R_3(j-1))$ such that $col(y)=col(x)$. Observe that $V(X_j)\setminus V(X_{j-1})$ is the set of vertices of $Y_j=HD[col(g_j)]$ that have not been identified in the $j^{th}$ iteration. So, for any vertex $x\in V(X_j)\setminus V(X_{j-1})$, $col(x)$ divides $col(g_j)$ but $col(x)$ does not divide $col(I(i,j))$ for any $i<j$. This means, for each such $x$, there exists $y$ in $V(R_3(j))\setminus V(R_3(j-1))$ with color $col(x)$ and this $y$ is unique since $V(R_3(j))\setminus V(R_3(j-1))$ contains the vertices of $Des(g_j)$ that are not descendant reachable from any $g_i$ where $i<j$. The uniqueness of colors in $Y_j=HD[col(g_j)]$ also implies that $f_j$ is a bijection.

  Now to show that $f_j$ is an isomorphism between $X_j$ and $R_3(j)$, it remains to show that $f_j$ preserves edge relations between $X_{j}$ and $R_3(j)$.

Here, we want to emphasize that it might happen that two vertices $x,x'$ in $X_{j-1}$ are not adjacent to each other, but after the vertex identification process in the $j^{th}$ iteration, there is an edge between $x$ and $x'$ in $X_j$. Moreover, through the following claim, we want to show that this incident has a correspondence in $R_3(j)$. 
  \begin{claim}\label{newly added edge}
    Let $x,x'$ be two vertices in the old graph (i.e., $X_{j-1}$) that take part in the vertex identification process in the $j^{th}$ iteration, i.e., $x, x' \in X_{j-1,j-1}$. Then, $\{x,x'\} \notin E(X_{j-1})$, but $\{x,x'\}\in E(X_j)$ if and only if $\{f_{j-1}(x),f_{j-1}(x')\} \notin E(R_3(j-1))$, but $\{f_j(x),f_j(x')\}\in E(R_3(j))$.
  \end{claim}
  \textbf{Proof of claim:}
  As $f_{j-1}$ is an isomorphism between $X_{j-1}$ and $R_3(j-1)$, we have $\{x,x'\} \notin E(X_{j-1})$ if and only if $\{f_j(x),f_j(x')\} \notin E(R_3(j-1))$.

Now, assume that $\{x,x'\}\notin E(X_{j-1})$ but $\{x,x'\}\in E(X_{j})$. Since $x,x' \in  X_{j-1,j-1}$, the vertices $x,x'$ get identified with some elements $z,z'$ respectively in $Y_j$
such that $\{z,z'\}\in E(Y_j)$. 
Also, $col(x)=col(z)=col(f_j(x))$ and $col(x')=col(z')=col(f_j(x'))$. Since $Y_j=HD[col(g_j)]$ and $\{z,z'\}\in E(Y_j)$, by definition either $col(z)=col(z')\cdot p$ or $col(z')=col(z)\cdot p$ for some prime $p$. Therefore, either $col(f_j(x))=col(f_j(x'))\cdot p$ or $col(f_j(x'))=col(f_j(x))\cdot p$ for some prime $p$. Moreover, $f_j(x),f_j(x')\in Des(g_j)$. Hence, by (2) of \Cref{Des(g_i)=HD[g_i]}, $\{f_j(x),f_j(x')\}\in E(R_3(j))$. 

Conversely, assume that $\{f_j(x),f_j(x')\}\in E(R_3(j))$. 
Since $x,x' \in X_{j-1,j-1}$, $x$ and $x'$ must have been identified with some vertices $z$ and $z'$ in $Y_j$ respectively such that $col(x)=col(z)$ and $col(x')=col(z')$. Now, because of (2) of \Cref{Des(g_i)=HD[g_i]}, $\{f_j(x),f_j(x')\} \in E(R_3(j))$ implies either $col(f_j(x))=col(f_j(x'))\cdot p$ or $col(f_j(x'))=col(f_j(x))\cdot p$ for some prime $p$. Therefore, either $col(z)=col(z')\cdot p$ or $col(z')=col(z)\cdot p$. So, $\{z,z'\}\in E(Y_j)$. Hence, after the vertex identification process, $\{x,x'\}\in E(X_3(j))$. \hfill $\Box$

    Now to show the preservation of edge relations, we consider the following cases, not necessarily disjoint:

(a) Let $x, x' \in V(X_{j-1})$, i.e., both the vertices are from the graph obtained in the previous iteration of the outer while loop. Then, by definition of $f_j$ in  \ref{defintion of f_j}, $f_j(x)=f_{j-1}(x)$ and $f_j(x')=f_{j-1}(x')$. Since by induction hypothesis $f_{j-1}$ is an isomorphism between $X_{j-1}$ and $R_3(j-1)$, $\{x,x'\} \in E(X_{j-1}) \ \iff \{f_{j-1}(x),f_{j-1}(x')\} \in E(R_3(j-1))$. The remaining case is covered by \Cref{newly added edge}. 


(b)  Let $x,x'$ be two vertices in $X_j$ that appear in the `$Y_j$-part' of $X_j$. More precisely, $x,x'$ belong to the disjoint union of 
$V(X_j)\setminus V(X_{j-1})$ ( which is the set of vertices which are newly introduced in the $j^{th}$ iteration of the outer while loop but not identified in the same ) and $X_{j-1,j-1}$ (which corresponds to the set of vertices that are the result of vertex identification of $X_{j-1,j-1}$ and $Y_{j,j-1}$ in the $j^{th}$ iteration). 
Since $Y_j=HD[col(g_j)]\cong R_3(G)[Des(g_j)]$ by \Cref{Des(g_i)=HD[g_i]}, $\{x,x'\}\in E(X_j)\iff \{f_j(x),f_j(x')\}\in E(R_3(j))$.

(c) Let $x$ be a vertex from the old graph $X_{j-1}$ which has not been identified in the $j^{th}$ iteration, i.e., $x \in V(X_{j-1}) \setminus X_{j-1,j-1}$. Let $x'$ be a newly added vertex which has not been identified in the $j^{th}$ iteration, i.e., $x' \in V(X_j)\setminus V(X_{j-1})$. It is not hard to see that $\{x,x'\}$ is not an edge of the disjoint union of $X_{j-1}$ and $Y_j$ ( before the identification process ). Since none of $x$ and $x'$ has taken part in the identification process in this iteration, we have $\{x,x'\}\notin E(X_j)$. Now as $f_j$ is a bijection, we also have the following: $f_j(x) \in V(R_3(j-1))\setminus Des(g_j)$ and $f_j(x') \in V(R_3(j))\setminus V(R_3(j-1))$. Since $f_j(x)$ and $f_j(x')$ are not in same $Des(u)$ for any vertex $u$ in $R_3(j)$, $\{f_j(x),f_j(x')\}$ is not an edge in $R_3(j)$. 
Thus, it is proved that $f_j$ is an isomorphism between $X_j$ and $R_3(j)$. So, we can conclude that $X_m \cong R_3(m)$. It is easy to see that $R_3(m)$ is $R_3(G)$. This concludes the proof of \Cref{X_j iso R3(j)}. \hfill $\Box$

Hence, the algorithm is correct and we can return an isomorphic copy of $R_3(G)$ from $R_4(G)$.
\end{proof}



\bibliography{References}
\appendix
\input{Appendix}

    



\end{document}

%% file: Appendix.tex
\newpage
\section{Appendix}


\subsection{Proof of Lemma \ref{isomorphism of strong product}}\label{proof of lemma 35}
\isoStrongProduct*
\begin{proof}\label{Proof: isomorphism of strong product}
 It is enough to prove the lemma for $k=2$. Let $f_1:V(DPow(G_1))\rightarrow V(DPow(H_1))$ and $f_2:V(DPow(G_2))\rightarrow V(DPow(H_2))$ be two isomorphisms from $DPow(G_1)$ to $DPow(H_1)$ and from $DPow(G_2)$ to $DPow(H_2)$ respectively. Let us define $f:V(DPow(G))\rightarrow  V(DPow(H))$ as $f((u_1,u_2))=(f_1(u_1),f_2(u_2))$. Since $f_1$ and $f_2$ are bijections, so is $f$. We show that $f$ preserves the edge relations between $DPow(G)$ and $DPow(H)$. Let us consider an edge $((u_1,u_2),(v_1,v_2))$ from $E(DPow(G))=E(DPow(G_1)\boxtimes DPow(G_2))$ (This equality follows from \Cref{directed power graph of direct product is strong product}.). Now from \Cref{Definition:strong product} and the facts that $f_1$ and $f_2$ are isomorphisms from $DPow(G_1)$ to $ DPow(H_1)$ and from $DPow(G_2)$ to $DPow(H_2)$ respectively, we have the following three scenarios:
    \begin{enumerate}
        \item $u_1=v_1$ and $(u_2,v_2)\in E(DPow(G_2))$. In this case, $f_1(u_1)=f_1(v_1)$ and $(f_2(u_2),f_2(v_2))\in E(DPow(H_2))$.
        \item $u_2=v_2$ and $(u_1,v_1)\in E(DPow(G_1))$. In this case, $f_2(u_2)=f_2(v_2)$ and $(f_1(u_1),f_1(v_1))\in E(DPow(H_1)$. 
        \item $(u_1,v_1)\in E(DPow(G_1))$ and $(u_2,v_2) \in E(DPow(G_2))$. In this case, $(f_1(u_1),f_1(v_1))\in E(DPow(H_1))$ and $(f_2(u_2),f_2(v_2))\in E(DPow(H_2))$.
    \end{enumerate}
    In all the three scenarios, by \Cref{Definition:strong product}, we have $((f_1(u_1),f_2(u_2)),(f_1(v_1),f_2(v_2))) \in E(Dpow(H_1)\boxtimes DPow(H_2))$. Therefore by \Cref{directed power graph of direct product is strong product}, $(f((u_1,u_2)),f((v_1,v_2))) \in E(DPow(H))$. 
 
    For the other direction, let $f:V(DPow(G))\rightarrow V(DPow(H))$ be an isomorphism between $DPow(G)$ and $DPow(H)$. Consider the sets $A_i=\{ (u,v) \in V(DPow(G)): out$-$deg((u,v))$ divides $|G_i| \}$ and $B_i=\{ (u',v') \in V(DPow(H)): out$-$deg((u',v'))$ divides $|H_i| \}$  for $i=1,2$. Recall that here the out-degree of a vertex is the order of the element and $o((u,v))=o(u)\cdot o(v)$. Since $|G_1|\times |G_2|=|G|$ and $gcd(|G_1|,|G_2|)=1$, it is easy to see that $A_i$ indeed corresponds to $V(DPow(G_i))$ for $i=1,2$. 
    Also, the subgraph of $DPow(G_1\times G_2)$ induced by $A_i$ corresponds to $DPow(G_i)$ for $i=1,2$. Similarly, we can see that $B_i$ corresponds to $V(DPow(H_i))$ and the subgraph induced by $B_i$ corresponds to $DPow(H_i)$ for $i=1,2$. Now the isomorphism $f$ preserves the out-degrees of the vertices. We denote the restriction of $f$ on $A_i$ by $f_i$. Then it is easy to see that $f_i$ is a bijection from $A_i$ to $B_i$. Also, there is only one element, namely the identity element, of out-degree $1$ (self-loop) and common in both $A_1$ and $A_2$. Also, that element is unique in $DPow(G)$.  One can see that $f_i:V(DPow(G_i))\rightarrow V(DPow(H_i))$ is an isomorphism between $DPow(G_i)$ and $DPow(H_i)$, for all $ i=1,2$.
   \end{proof}


        


\subsection{Proof of Theorem \ref{complete power graph}}\label{Proof of Theorem 18}
\powCliquelemma*
\begin{proof}

    If $G=\mathbb{Z}_{p^m}$, then $\Gamma$ is complete.

    For the other direction, if $G$ is not cyclic then the cyclic cover of $G$ has at least two maximal cyclic subgroups $\langle g_1 \rangle$ and $\langle g_2 \rangle$. If $x$ and $y$ are generators of $\langle g_1 \rangle$ and $\langle g_2 \rangle$ respectively, then they are not adjacent. Therefore, we can assume that $G=\mathbb{Z}_m$ for some $m$ such that $m$ is not a prime power. Now, let $v$ be a generator of $\mathbb{Z}_m$. Let $u_1$ and $u_2$ be two non-generator elements of $G$ with different orders. Then, by \Cref{generators are the only twins}, $u_1$ and $u_2$ are not closed-twins in $\Gamma_v=\Gamma$. Therefore, $\Gamma$ is not complete.
\end{proof}

%% file: main.bbl
\begin{thebibliography}{10}

\bibitem{aho1974design}
Alfred~V. Aho and John~E. Hopcroft.
\newblock {\em The design and analysis of computer algorithms}.
\newblock Pearson Education India, 1974.

\bibitem{arvind2022recognizing}
V.~Arvind and Peter~J. Cameron.
\newblock Recognizing the commuting graph of a finite group.
\newblock {\em arXiv preprint arXiv:2206.01059}, 2022.

\bibitem{babai2016graph}
L{\'a}szl{\'o} Babai.
\newblock Graph isomorphism in quasipolynomial time.
\newblock In {\em Proceedings of the forty-eighth annual ACM symposium on
  Theory of Computing}, pages 684--697, 2016.

\bibitem{babai1982isomorphism}
L{\'a}szl{\'o} Babai, D.~Yu Grigoryev, and David~M. Mount.
\newblock Isomorphism of graphs with bounded eigenvalue multiplicity.
\newblock In {\em Proceedings of the fourteenth annual ACM symposium on Theory
  of computing}, pages 310--324, 1982.

\bibitem{babai1983canonical}
L{\'a}szl{\'o} Babai and Eugene~M Luks.
\newblock Canonical labeling of graphs.
\newblock In {\em Proceedings of the fifteenth annual ACM symposium on Theory
  of computing}, pages 171--183, 1983.

\bibitem{bodlaender1990polynomial}
Hans~L. Bodlaender.
\newblock Polynomial algorithms for graph isomorphism and chromatic index on
  partial k-trees.
\newblock {\em Journal of Algorithms}, 11(4):631--643, 1990.

\bibitem{boppana1987does}
Ravi~B. Boppana, Johan Hastad, and Stathis Zachos.
\newblock Does co-np have short interactive proofs?
\newblock {\em Information Processing Letters}, 25(2):127--132, 1987.

\bibitem{bubboloni2022critical}
Daniela Bubboloni and Nicolas Pinzauti.
\newblock Critical classes of power graphs and reconstruction of directed power
  graphs.
\newblock {\em arXiv preprint arXiv:2211.14778}, 2022.

\bibitem{cameron2010power}
Peter~J. Cameron.
\newblock The power graph of a finite group, ii.
\newblock {\em Journal of Group Theory}, 13(6):779--783, 2010.

\bibitem{cameron2021graphs}
Peter~J. Cameron.
\newblock Graphs defined on groups.
\newblock {\em International Journal of Group Theory}, 11(2):53--107, 2022.

\bibitem{cameron2011power}
Peter~J. Cameron and Shamik Ghosh.
\newblock The power graph of a finite group.
\newblock {\em Discrete Mathematics}, 311(13):1220--1222, 2011.

\bibitem{chakrabarty2009undirected}
Ivy Chakrabarty, Shamik Ghosh, and M.~K. Sen.
\newblock Undirected power graphs of semigroups.
\newblock In {\em Semigroup Forum}, volume~78, pages 410--426. Springer, 2009.

\bibitem{feng2016full}
Min Feng, Xuanlong Ma, and Kaishun Wang.
\newblock The full automorphism group of the power (di) graph of a finite
  group.
\newblock {\em European Journal of Combinatorics}, 52:197--206, 2016.

\bibitem{grochow2021p}
Joshua~A. Grochow and Youming Qiao.
\newblock On p-group isomorphism: search-to-decision, counting-to-decision, and
  nilpotency class reductions via tensors.
\newblock In {\em 36th Computational Complexity Conference (CCC 2021)}, volume
  200, 2021.

\bibitem{grohe2019linear}
Martin Grohe and Sandra Kiefer.
\newblock A linear upper bound on the weisfeiler-leman dimension of graphs of
  bounded genus.
\newblock In {\em 46th International Colloquium on Automata, Languages, and
  Programming (ICALP 2019)}. Schloss Dagstuhl-Leibniz-Zentrum fuer Informatik,
  2019.

\bibitem{grohe2021isomorphism}
Martin Grohe and Daniel Neuen.
\newblock Isomorphism, canonization, and definability for graphs of bounded
  rank width.
\newblock {\em Communications of the ACM}, 64(5):98--105, 2021.

\bibitem{grohe2020faster}
Martin Grohe, Daniel Neuen, and Pascal Schweitzer.
\newblock A faster isomorphism test for graphs of small degree.
\newblock {\em SIAM Journal on Computing}, pages FOCS18--1, 2020.

\bibitem{grohe2020improved}
Martin Grohe, Daniel Neuen, Pascal Schweitzer, and Daniel Wiebking.
\newblock An improved isomorphism test for bounded-tree-width graphs.
\newblock {\em ACM Transactions on Algorithms (TALG)}, 16(3):1--31, 2020.

\bibitem{grohe2015isomorphism}
Martin Grohe and Pascal Schweitzer.
\newblock Isomorphism testing for graphs of bounded rank width.
\newblock In {\em 2015 IEEE 56th Annual Symposium on Foundations of Computer
  Science}, pages 1010--1029. IEEE, 2015.

\bibitem{hellmuth2015cartesian}
Marc Hellmuth and Tilen Marc.
\newblock On the cartesian skeleton and the factorization of the strong product
  of digraphs.
\newblock {\em Theoretical Computer Science}, 565:16--29, 2015.

\bibitem{imrich2008topics}
Wilfried Imrich, Sandi Klavzar, and Douglas~F. Rall.
\newblock {\em Topics in graph theory: Graphs and their Cartesian product}.
\newblock CRC Press, 2008.

\bibitem{kelarev2000combinatorial}
Andrei~V. Kelarev and Stephen~J. Quinn.
\newblock A combinatorial property and power graphs of groups.
\newblock {\em Contributions to general algebra}, 12(58):3--6, 2000.

\bibitem{kumar2021recent}
Ajay Kumar, Lavanya Selvaganesh, Peter~J. Cameron, and T.~Tamizh Chelvam.
\newblock Recent developments on the power graph of finite groups--a survey.
\newblock {\em AKCE International Journal of Graphs and Combinatorics},
  18(2):65--94, 2021.

\bibitem{luks1982isomorphism}
Eugene~M. Luks.
\newblock Isomorphism of graphs of bounded valence can be tested in polynomial
  time.
\newblock {\em Journal of computer and system sciences}, 25(1):42--65, 1982.

\bibitem{mckenzie1971cardinal}
Ralph McKenzie.
\newblock Cardinal multiplication of structures with a reflexive relation.
\newblock {\em Fundamenta Mathematicae}, 70(1):59--101, 1971.

\bibitem{miller1980isomorphism}
Gary Miller.
\newblock Isomorphism testing for graphs of bounded genus.
\newblock In {\em Proceedings of the twelfth annual ACM symposium on Theory of
  computing}, pages 225--235, 1980.

\bibitem{miller1978nlog}
Gary~L. Miller.
\newblock On the nlog n isomorphism technique (a preliminary report).
\newblock In {\em Proceedings of the tenth annual ACM symposium on theory of
  computing}, pages 51--58, 1978.

\bibitem{mukherjee2019hamiltonian}
Himadri Mukherjee.
\newblock Hamiltonian cycles of power graph of abelian groups.
\newblock {\em Afrika Matematika}, 30:1025--1040, 2019.

\bibitem{oxley2006matroid}
J.~G. Oxley.
\newblock {\em Matroid Theory}.
\newblock Oxford graduate texts in mathematics. Oxford University Press, 2006.

\bibitem{rotman2012introduction}
Joseph~J. Rotman.
\newblock {\em An introduction to the theory of groups}, volume 148.
\newblock Springer Science \& Business Media, 2012.

\bibitem{west2001introduction}
Douglas~B. West.
\newblock {\em Introduction to Graph Theory}.
\newblock Prentice Hall, September 2000.

\bibitem{wiebking2020graph}
Daniel Wiebking.
\newblock Graph isomorphism in quasipolynomial time parameterized by treewidth.
\newblock In {\em 47th International Colloquium on Automata, Languages, and
  Programming (ICALP 2020)}. Schloss Dagstuhl-Leibniz-Zentrum f{\"u}r
  Informatik, 2020.

\bibitem{zahirovic2020study}
Samir Zahirovi{\'c}, Ivica Bo{\v{s}}njak, and Roz{\'a}lia Madar{\'a}sz.
\newblock A study of enhanced power graphs of finite groups.
\newblock {\em Journal of Algebra and Its Applications}, 19(04):2050062, 2020.

\end{thebibliography}
